\documentclass[12pt,a4paper]{article}       

\usepackage{amsthm}
\usepackage{amsmath}
\usepackage{graphicx}
\usepackage{amssymb}
\usepackage{bold-extra}

\newtheorem{thm}{Theorem}[section]
\newtheorem{lma}[thm]{Lemma}
\newtheorem{cor}[thm]{Corollary}

\newtheorem{prob}{Problem}
\newtheorem{prop}[thm]{Proposition}

\newtheorem{case}{Case}
\newtheorem{subcase}{Case}[case]

\DeclareMathOperator{\col}{col}
\DeclareMathOperator{\bare}{trunk}

\begin{document}

\title{The complexity of FREE FLOOD IT on $2 \times n$ boards}
\date{\today}
\author{Kitty Meeks and Alexander Scott\\
\small{Mathematical Institute, University of Oxford, 24-29 St Giles, Oxford OX1 3LB, UK} \\
\texttt{\small{\{meeks,scott\}@maths.ox.ac.uk}}}
\maketitle

\begin{abstract}
We consider the complexity of problems related to the combinatorial game Free-Flood-It, in which players aim to make a coloured graph monochromatic with the minimum possible number of flooding operations.  Our main result is that computing the length of an optimal sequence is fixed parameter tractable (with the number of colours as a parameter) when restricted to rectangular $2 \times n$ boards.  We also show that, when the number of colours is unbounded, the problem remains NP-hard on such boards.  These results resolve a question of Clifford, Jalsenius, Montanaro and Sach.
\end{abstract}

\section{Introduction}
In this paper we consider the complexity of problems related to the one-player combinatorial game Flood-It, introduced by Arthur, Clifford, Jalsenius, Montanaro and Sach in \cite{arthurFUN}.  The original game is played on a board consisting of an $n \times n$ grid of coloured squares, each square given a colour from some fixed colour-set, but we can more generally regard the game as being played on a vertex-coloured graph.  A move then consists of picking a vertex $v$ and a colour $d$, and giving all vertices in the same monochromatic component as $v$ colour $d$.  The goal is to make the entire graph monochromatic with as few such moves as possible.

When the game is played on a planar graph, it can be regarded as modelling repeated use of the flood-fill tool in Microsoft Paint.  Implementations of the game, played on a square grid, are widely available online, and include a flash game \cite{flash} as well as popular smartphone apps \cite{iphoneapp,androidapp}.  There also exist implementations using a hexagonal grid: Mad Virus \cite{madvirus} is the same one-player game described above, while the Honey Bee Game \cite{honeybee} is a two player variant, and has been studied by Fleischer and Woeginger \cite{fleischer10}.  All these implementations are based on the ``fixed'' version of the game, where all moves must be played at the same fixed vertex (usually the vertex corresponding to the top left square when the board is an $n \times n$ grid).

For any coloured graph, we define the following problems.
\begin{itemize}
\item \textsc{Free-Flood-It} is the problem of determining the minimum number of moves required to flood the graph, if we are allowed to make moves anywhere in the graph.
\item \textsc{Fixed-Flood-It} is the same problem when all moves must be played at a single specified vertex.\footnote{\textsc{Fixed Flood It} is often referred to as simply \textsc{Flood-It}, but we use the longer name to avoid confusion with the free version.}
\item $c$-\textsc{Free-Flood-It} and $c$-\textsc{Fixed-Flood-It} respectively are the variants of \textsc{Free-Flood-It} and \textsc{Fixed-Flood-It} in which only colours from some fixed set of size $c$ are used.
\end{itemize}
Note that we can trivially flood an $n$-vertex graph with $n-1$ moves, and that if $c$ colours are present in the initial colouring we require at least $c-1$ moves.

These problems are known to be computationally difficult in many situations.  In \cite{arthurFUN}, Arthur, Clifford, Jalsenius, Montanaro and Sach proved that $c$-\textsc{Free-Flood-It} is NP-hard in the case of an $n \times n$ grid, for every $c \geq 3$, and that this result also holds for the fixed variant.  Lagoutte, Noual and Thierry \cite{lagoutte,lagoutte11} showed that the same result holds when the game is played instead on a hexagonal grid, as in Mad Virus or a one-player version of the Honey Bee Game.  Fleischer and Woeginger \cite{fleischer10} proved that $c$-\textsc{Fixed Flood It} remains NP-hard when restricted to trees, for every $c \geq 4$,\footnote{Note that this proof does in fact require four colours, not three as stated in a previous version of \cite{fleischer10}.} and Fukui, Nakanishi, Uehara, Uno and Uno \cite{fukui} demonstrated that this result can be extended to show the hardness $c$-\textsc{Free Flood It} under the same conditions.

A few positive results are known, however.  2-\textsc{Free-Flood-It} is solvable in polynomial time on arbitrary graphs, a result shown independently by Clifford et.~al.~\cite{clifford}, Lagoutte \cite{lagoutte} and Meeks and Scott \cite{general}.  It is also known that \textsc{Fixed-Flood-It} and \textsc{Free-Flood-It} are solvable in polynomial time on paths \cite{clifford,general,fukui} and cycles \cite{fukui}, and more generally on any graph with only a polynomial number of connected subgraphs \cite{spanningFUN,spanning}.  Meeks and Scott also show that the number of moves required to create a monochromatic component containing an arbitrary, bounded-size subset of the vertices can be computed in polynomial time, even when the number of colours is unbounded \cite{spanning,spanningFUN}.

A major focus of previous research has been the restriction of the game to rectangular boards of fixed height.  Although an additive approximation for $c$-\textsc{Free-Flood-It} can be computed in polynomial time \cite{general}, solving either $c$-\textsc{Free-Flood-It} or $c$-\textsc{Fixed-Flood-It} exactly remains NP-hard on $3 \times n$ boards, whenever $c \geq 4$ \cite{general}.  However, Clifford et.~al.~\cite{clifford} give a linear time algorithm for \textsc{Fixed-Flood-It} on $2 \times n$ boards.   They also raise the question of the complexity of the free variant in this setting.

Here we address this remaining case of ($c$-)\textsc{Free-Flood-It} restricted to $2 \times n$ boards, which turn out to be a particularly interesting class of graphs on which to analyse the game.  The majority of the paper describes an algorithm to demonstrate that $c$-\textsc{Free-Flood-It}, restricted to $2 \times n$ boards, is fixed parameter tractable with parameter $c$.  To do this we exploit some general results from \cite{spanning} about the relationship between the number of moves required to flood a graph and its spanning trees.  

On the other hand, we also show that \textsc{Free-Flood-It} remains NP-hard in this setting.  This is a somewhat surprising result, as it gives the first example of a class of graphs on which the complexity of \textsc{Fixed-Flood-It} and \textsc{Free-Flood-It} has been shown to be different.

The rest of the paper is organised as follows.  We begin with notation and definitions in Section \ref{notation}, before giving our algorithm for $c$-\textsc{Free-Flood-It} in Section \ref{fpt}.  Finally, in Section \ref{NPhard}, we show that the problem remains NP-hard when the number of colours used is unbounded.

\section{Notation and definitions}
\label{notation}

Although the original Flood-It game is played on a square grid, and our main results here concern the game restricted to a rectangular grid, it is convenient to consider the generalisation of the game to an arbitrary graph $G=(V,E)$, equipped with an initial colouring $\omega$ using colours from the \emph{colour-set} $C$.  Then each move $m=(v,d)$ consists of choosing some vertex $v \in V$ and a colour $d \in C$, and assigning colour $d$ to all vertices in the same monochromatic component as $v$.  The goal is to give every vertex in $G$ the same colour, using as few moves as possible. 

Given any connected graph $G$, equipped with a colouring $\omega$ (not necessarily proper), we define $m(G,\omega,d)$ to be the minimum number of moves required in the free variant to give all its vertices colour $d$, and $m(G,\omega)$ to be $\min_{d \in C}m(G,\omega,d)$.  If $S$ is a sequence of moves played on a graph $G$ with initial colouring $\omega$, we denote by $S(\omega,G)$ the new colouring obtained by playing $S$ in $G$.  Note that, if the initial colouring $\omega$ of $G$ is not proper, we may obtain an equivalent coloured graph $G'$ (with colouring $\omega'$) by contracting monochromatic components of $G$ with respect to $\omega$.

Let $A$ be any subset of $V$.  We denote by $\col(A,\omega)$ the set of colours assigned to vertices of $A$ by $\omega$. We say a move $m = (v,d)$ is \emph{played in} $A$ if $v \in A$, and that $A$ is \emph{linked} if it is contained in a single monochromatic component.  Subsets $A,B \subseteq V$ are \emph{adjacent} if there exists $ab \in E$ with $a \in A$ and $b \in B$.    

When we consider the game played on a rectangular board $B$, we are effectively playing the game in a corresponding coloured graph $G$, obtained from the planar dual of $B$ (in which there is one vertex corresponding to each square of $B$, and vertices are adjacent if they correspond to squares which are either horizontally or vertically adjacent in $B$) by giving each vertex the colour of the corresponding square in $B$.  We identify areas of $B$ with the corresponding subgraphs of $G$, and may refer to them interchangeably.

We define a \emph{border} of $B$ to be a union of edges of squares on the original board $B$ that forms a path from the top edge of the board to the bottom (but not including any edges that form the top or bottom edge of the board).  Thus, a border in $B$ corresponds to an edge-cut in the corresponding graph.  Observe that a border is uniquely defined by the points at which it meets the top and bottom of the board, so there are $(n+1)^2$ borders in total.  We denote by $b_L$ and $b_R$ the borders corresponding to the left-hand and right-hand edges of the board respectively.  Given two borders $b_1$ and $b_2$, we write $b_1 \leq b_2$ if and only if $b_1$ meets both the top and bottom of the board to the left of (or at the same point as) $b_2$, and write $b_1 < b_2$ if $b_1 \leq b_2$ and $b_1 \neq b_2$.  Note that if $b_1 \leq b_2$ then $b_1$ lies entirely to the left of $b_2$ (the two borders may meet but never cross); this is a special property of $2 \times n$ boards and does not hold for edge-cuts in graphs corresponding to $k \times n$ boards for $k \geq 3$.  

If $G$ is the graph corresponding to the $2 \times n$ board $B$, we say that a vertex (or subgraph) is \emph{incident} with a border $b$ if the vertex (or some vertex in the subgraph) corresponds to a square on $B$ whose edge forms part of $b$.  If $b_1 < b_2$ are borders, we denote the subgraph induced by vertices lying between $b_1$ and $b_2$ by $B[b_1,b_2]$, and we say $B[b_1,b_2]$ is a \emph{section} if it is connected.  

Finally, given any tree $T$, we denote by $\bare(T)$ the subtree obtained by deleting all leaves of $T$, and given any $x,y \in V(T)$ we set $P(T,x,y)$ to be the unique path from $x$ to $y$ in $T$.

\section{$c$-\textsc{Free-Flood-It} on $2 \times n$ boards}
\label{fpt}

In this section, we give an algorithm to solve $c$-\textsc{Free-Flood-It} on $2 \times n$ boards.  More specifically, we prove the following result, which shows that $c$-\textsc{Free-Flood-It}, restricted to $2 \times n$ boards, is fixed parameter tractable, parameterised by $c$.  This answers an open question of Clifford, Jalsenius, Montanaro and Sach \cite{clifford}.

\begin{thm}
When restricted to $2 \times n$ boards, $c$-\textsc{Free-Flood-It} can be solved in time $O(n^{10} \cdot 2^{c})$.
\label{2xn-fpt}
\end{thm}

We begin with some background and auxiliary results in Section \ref{background}, and then describe the algorithm in Section \ref{algorithm}.  The algorithm is based on the fact, proved in Section \ref{background}, that if $G$ (with colouring $\omega$) is a graph corresponding to a $2 \times n$ board, then $G$ has a spanning tree $T$ such that
\begin{enumerate}
\item $\bare(T)$ is a path,
\item $m_T(T,\omega) = m_G(G,\omega)$, and
\item there is an optimal sequence to flood $T$ in which all moves are played in $\bare(T)$.
\end{enumerate}
We make use of this fact to show that, in order to compute the number of moves required to flood $G$, we can instead consider the number of moves required to flood some appropriately chosen paths (while keeping track of the effect that flooding the paths has on vertices that do not lie on the path).

\subsection{Background and auxiliary results}
\label{background}

Before describing our algorithm in the next section, we need a number of results which will be used to prove its correctness.  We begin with some previous results from \cite{spanning}.  Meeks and Scott prove that it suffices to consider spanning trees in order to determine the minimum number of moves required to flood a graph.  For any connected graph $G$, let $\mathcal{T}(G)$ denote the set of all spanning trees of $G$. 

\begin{thm}
Let $G$ be a connected graph with colouring $\omega$ from colour-set $C$.  Then, for any $d \in C$,
$$m(G,\omega,d) = \min_{T \in \mathcal{T}(G)} m(T,\omega,d).$$
\label{spanning-tree}
\end{thm}
For any $d \in C$, we say that $T$ is a \emph{$d$-minimal} spanning tree for $G$  (with respect to $\omega$) if $m(T,\omega,d) = m(G,\omega,d)$.

In the remainder of this section, we prove that in the special case in which $G$ corresponds to a $2 \times n$ board, there is always a $d$-minimal spanning tree $T$ such that $\bare(T)$ is a path.

In doing so, and in proving the correctness of our algorithm in the next section, we make use of a corollary of Theorem \ref{spanning-tree}, again proved in \cite{spanning}, which shows that the number of moves required to flood a graph is bounded above by the sum of the numbers of moves required to flood connected subgraphs which cover the vertex-set.

\begin{cor}
Let $G$ be a connected graph, with colouring $\omega$ from colour-set $C$, and let $A$ and $B$ be subsets of $V(G)$ such that $V(G) = A \cup B$ and $G[A], G[B]$ are connected.  Then, for any $d \in C$,
$$m(G,\omega,d) \leq m(A,\omega,d) + m(B,\omega,d).$$
\label{non-interference}
\end{cor}

A key step used to prove Theorem \ref{spanning-tree} in \cite{spanning} is to prove a special case of Corollary \ref{non-interference}, where the underlying graph $G$ is a tree and $A$ and $B$ are disjoint.  We will need the following result, proved using an extension of part of this proof from \cite{spanning}.

\begin{lma}
Let $T$ be a tree, with colouring $\omega$ from colour-set $C$, let $A$ and $B$ be disjoint subsets of $V(T)$ such that $V(T) = A \cup B$ and $T[A], T[B]$ are connected, and let $x$ be the unique vertex of $B$ with a neighbour in $A$.  Suppose that 
\begin{itemize}
\item the sequence $S_A$ floods $T[A]$ with colour $d_A$,
\item the sequence $S_B$ floods $T[B]$ with colour $d_B$,  
\item at least one move of $S_B$ changes the colour of $x$, and
\item  playing $S_A$ in $T$ changes the colour of $x$.
\end{itemize}
Then
$$m(T,\omega,d_B) \leq |S_A| + |S_B|.$$
\label{compatible-ordering}
\end{lma}
\begin{proof}
We proceed by induction on $|B|$.  Note that we may assume without loss of generality that $\omega$ gives a proper colouring of $B$; otherwise we may contract monochromatic components.  Suppose $|B| = 1$.  Then $S_A$ must change the colour of the only vertex in $B$ (linking it to some $a \in A$), and so playing $S_A$ in $T$ makes the whole tree monochromatic with colour $d_A$.  Thus $m(T,\omega,d_A) \leq |S_A|$, and
$$m(T,\omega,d_B) \leq m(T,\omega,d_A) + 1 \leq |S_A| + 1 \leq |S_A| + |S_B|,$$
as required, since by assumption $|S_B| \geq 1$.

Now suppose $|B| > 1$, so $B$ is not monochromatic initially, and assume that the result holds for smaller $B$.  Set ${S_B}'$ to be the initial segment of $S_B$, up to and including the move that first makes $B$ monochromatic (in any colour $d'$), so any final moves that simply change the colour of $B$ are omitted.  We may, of course, have ${S_B}' = S_B$ (and so $d' = d_B$), if $B$ is not monochromatic before the final move of $S_B$.

First suppose that $S_B'$ does not change the colour of $x$ (which is only possible in the case $S_B' \neq S_B$).  Then playing $S_B'$ in $T$ to make $B$ monochromatic cannot change the colour of any vertex in $A$, so if we play $S_B'$ in $T$ and then play $S_A$, this will still flood $A$ with colour $d_A$.  Moreover, as playing $S_B'$ has not changed the colour of $x$, playing $S_A$ will still change the colour of $x$, thus linking all of $B$ to $A$ and so flooding $T$ with colour $d_A$.  Hence, in this case, we have 
$$m(T,\omega,d_A) \leq |S_B'| + |S_A|,$$
and so, as we must in this case have $|S_B'| < |S_B|$,
$$m(T,\omega,d_B) \leq 1 + m(T,\omega,d_A) \leq 1 + |S_B'| + |S_A| \leq |S_A| + |S_B|,$$
as required. 

Suppose now that $S_B'$ does change the colour of $x$.  Before the final move of ${S_B}'$ there are $r \geq 2$ monochromatic components in $B$ (all but one of which have colour $d'$), with vertex-sets $B_1, \ldots, B_r$.  For $1 \leq i \leq r$, set $S_i$ to be the subsequence of ${S_B}'$ consisting of moves played in $B_i$, and note that these subsequences partition ${S_B}'$.  Observe also that playing $S_i$ in $T[B_i]$ gives $B_i$ colour $d'$, so $m(B_i,\omega,d') \leq |S_i|$.

Let $B_1$ be the unique component adjacent to $A$, and set $T_1 =  T[A \cup B_1]$.  Note that $S_A$ floods $T_1[A]$ with colour $d_A$, and $S_1$ floods $T_1[B_1]$ with colour $d'$.  Moreover, as playing $S_A$ in $T$ changes the colour of $x$, playing $S_A$ in $T_1$ must also change the colour of $x$.  Also, at least one move from $S_B$ changes the colour of $x$, the unique vertex of $B_1$ with a neighbour in $A$, and this move must belong to $S_1$.  Thus we can apply the inductive hypothesis to see that 
$$m(T_1, \omega, d') \leq |S_A| + |S_1|.$$
Now suppose without loss of generality that $B_2$ is adjacent to $B_1$.  We can then apply Corollary \ref{non-interference} to $T_2 = T[V(T_1) \cup B_2]$ to see that
$$m(T_2, \omega, d') \leq m(T_1,\omega,d') + m(B_2,\omega,d') \leq |S_A| + |S_1| + |S_2|.$$
Continuing in this way, each time adding an adjacent component, we see that
$$m(T,\omega,d') \leq |S_A| + \sum_{i=1}^r |S_i| = |S_A| + |{S_B}'|.$$
Now, if ${S_B}' = S_B$, this immediately gives the desired result, as $d' = d_B$.  Otherwise, note that $|S_B| \geq |{S_B}'|+1$ and so
$$m(T,\omega,d_B) \leq m(T,\omega,d') + 1 \leq |S_A| + |{S_B}'| + 1 \leq |S_A| + |S_B|,$$
as required.
\end{proof}

In the next result, we exploit this lemma to give a strengthening of Corollary \ref{non-interference} under additional assumptions.  This can be applied to show that, in certain situations, we may assume that \emph{no} optimal sequence to flood a subtree can change the colour of any vertex outside the subtree, when played in a larger tree.

\begin{prop}
Let $T$ be a tree, with colouring $\omega$ from colour-set $C$, and let $X$ and $Y$ be vertex-disjoint sutbrees of $T$ such that $T[V(X) \cup V(Y)]$ is connected, and such that
\begin{itemize}
\item there is a sequence $S_X$ of $\alpha$ moves that floods $X$ with some colour $d' \in C$,
\item there is a sequence $S_Y$ of $\beta$ moves that floods $Y$ with colour $d$, and that changes the colour of the unique neighbour of $X$ in $Y$, and
\item playing $S_X$ in $T$ changes the colour of at least one vertex in $Y$.
\end{itemize}
Then, setting $T' =  T \setminus (V(X) \cup V(Y))$, we have 
$$m(T,\omega,d) \leq m(T',\omega',d) + \alpha + \beta.$$
\label{strong-non-int}
\end{prop}
\begin{proof}
Note that $S_X$ must change the colour of $v$, so we can apply Lemma \ref{compatible-ordering} to see that
$$m(T[V(X) \cup V(Y)],\omega,d) \leq |S_X| + |S_Y| = \alpha + \beta.$$
Corollary \ref{non-interference} then gives
\begin{align*}
m(T,\omega,d) & \leq m(T[V(X) \cup V(Y)], \omega,d) + m(T',\omega',d) \\
              & \leq m(T',\omega',d) + \alpha + \beta,
\end{align*}
as required.
\end{proof}

Before proving the main result of this section, we need one further result, relating the number of moves required to flood the same graph with different initial colourings.

\begin{lma}
Let $G$ be a connected graph, and let $\omega$ and $\omega'$ be two colourings of the vertices of $G$ (from colour-set $C$).  Let $\mathcal{A}$ be the set of all maximal monochromatic components of $G$ with respect to $\omega'$, and for each $A \in \mathcal{A}$ let $c_A$ be the colour of $A$ under $\omega'$.  Then, for any $d \in C$,
$$m(G,\omega,d) \leq m(G,\omega',d) + \sum_{A \in \mathcal{A}} m(A,\omega,c_A).$$
\label{change-colouring}
\end{lma}
\begin{proof}
We proceed by induction on $m(G,\omega',d)$.  Note that if $m(G,\omega',d) = 0$ then the result is trivially true: in this case $\mathcal{A}$ contains a single monochromatic component $G$, with colour $d$, so we have
$$m(G,\omega',d) + \sum_{A \in \mathcal{A}} m(A,\omega,c_A) = m(G,\omega,d).$$

Suppose now that $m(G,\omega',d) > 0$, and let $S$ be an optimal sequence of moves to flood $G$ with colour $d$, when the initial colouring is $\omega'$.  We proceed by case analysis on the final move, $\alpha$, of $S$.  First suppose that $G$ is already monochromatic before $\alpha$, so this final move just changes the colour of the entire graph to $d$ from some colour $d' \in C$.  In this case, $m(G,\omega',d) = m(G,\omega',d') + 1$, and so we may apply the inductive hypothesis to see that 
\begin{align*}
m(G,\omega,d) & \leq 1 + m(G,\omega,d') \\
              & \leq 1 + m(G,\omega',d') + \sum_{A \in \mathcal{A}} m(A,\omega,c_A) \\
              & = m(G,\omega',d) + \sum_{A \in \mathcal{A}} m(A,\omega,c_A),
\end{align*}
as required.

Now suppose that $G$ is not monochromatic before $\alpha$, and so this move links monochromatic components $X_1, \ldots, X_r$.  We may assume that $\alpha$ changes the colour of $X_1$ from $d'$ to $d$, and that all the components $X_2, \ldots, X_r$ have colour $d$ before $\alpha$.  Let $S_i$ denote the subsequence of $S$ consisting of moves played in $X_i$, and observe that playing $S_i$ in the isolated subgraph $X_i$ must flood this graph with colour $d$, so $m(X_i,\omega',d) \leq |S_i|$.  Note that, as no move can split a monochromatic component, the sets $\mathcal{A}_i = \{A \in \mathcal{A}: A \subseteq X_i\}$ (for $1 \leq i \leq r$) partition $\mathcal{A}$.

Observe that, for $2 \leq i \leq r$, $m(X_i, \omega, d) < |S| = m(G,\omega',d)$, and so we may apply the inductive hypothesis to see that
\begin{align*}
m(X_i,\omega,d) & \leq m(X_i,\omega',d) + \sum_{A \in \mathcal{A}_i} m(A,\omega,c_A) \\
                              & \leq |S_i| + \sum_{A \in \mathcal{A}_i} m(A,\omega,c_A).
\end{align*}                              
Similarly, the inductive hypothesis gives
$$m(X_1,\omega,d') \leq m(X_1,\omega',d') + \sum_{A \in \mathcal{A}_1} m(A,\omega,c_A),$$
and so, as $m(X_1,\omega',d') \leq |S_1| - 1$, we have
\begin{align*}
m(X_1,\omega,d) & \leq 1 + m(X_1,\omega,d) \\
                              & \leq |S_1| + \sum_{A \in \mathcal{A}_1} m(A,\omega,c_A).
\end{align*}
Now we can apply Corollary \ref{non-interference} to see that
$$m(G,\omega,d) \leq \sum_{i=1}^r m(X_i,\omega,d),$$
and so
\begin{align*}
m(G,\omega,d) & \leq \sum_{i=1}^r (|S_i| + \sum_{A \in \mathcal{A}_i} m(A,\omega,c_A)) \\
                 & = |S| + \sum_{A \in \mathcal{A}} m(A,\omega,c_A) \\
                 & = m(G,\omega',d) + \sum_{A \in \mathcal{A}} m(A,\omega,c_A),
\end{align*}
completing the proof.
\end{proof}

Using the previous results, we are now ready to prove the key result of this section. 

\begin{lma}
Let $G$ with colouring $\omega$ (from colour-set $C$) be the graph corresponding to a $2 \times n$ flood-it board $B$, let $H$ be a connected induced subgraph of $G$, and let $u$ and $w$ be vertices lying in the leftmost and rightmost columns of $H$ respectively.  Then, for any $d \in C$, there exists a $d$-minimal spanning tree $T$ for $H$ such that $\bare(T) \subseteq P(T,u,w)$.
\label{leafy-path}
\end{lma}
\begin{proof}
We proceed by induction on $m_H(H,\omega,d)$.  Note that the result is trivially true if $m_H(H,\omega,d) = 0$ as the graph is initially monochromatic with colour $d$ and so any spanning tree will do.  Suppose then that $m_H(H,\omega,d) > 0$.  Let $S$ be an optimal sequence to flood $H$ with colour $d$, and suppose that the last move of $S$ is $\alpha$.

If $H$ is monochromatic in some colour $d' \in C$ before $\alpha$ is played, and so this final move just changes the colour of the whole graph to $d$, we see that $m_H(H,\omega,d') \leq m_H(H,\omega,d) - 1$.  Thus we may apply the inductive hypothesis to obtain a $d'$-minimal spanning tree $T$ for $H$ such that $\bare(T) \subseteq P(T,u,w)$.  But then
$$m_T(T,\omega,d) \leq 1 + m_T(T,\omega,d') = 1 + m_H(H,\omega,d') \leq m_H(H,\omega,d),$$
and so $T$ is also a $d$-minimal spanning tree for $H$.

Thus we may assume that $H$ is not monochromatic immediately before $\alpha$ is played.  This means that $\alpha$ must change the colour of a monochromatic component $A$ from some $d' \in C$ to $d$, where $H \setminus A$ is nonempty and has colour $d$ before $\alpha$ is played.  Since $H$ is a connected induced subgraph of a $2 \times n$ board, $H \setminus A$ has at most one component $L$ which contains vertices lying in columns to the left of all columns containing a vertex of $A$, and correspondingly at most one component $R$ containing vertices lying in columns entirely to the right of $A$.  There may additionally be some components $X_1, \ldots, X_r$ of $H \setminus A$ which contain only vertices which lie in the same column as some vertex of $A$.  A possible structure for $H$ is illustrated in Figure \ref{H-cpts}.  We will exploit the structure of $H \setminus A$ to define a $d$-minimal spanning tree $T$ for $H$ whose non-leaf vertices lie on $P(T,u,w)$.

The remainder of the proof is structured as follows.  We begin by defining a spanning tree $T$ for $H$ and identifying certain important substructures in $T$, and then go on to show that $T$ is in fact $d$-minimal.  To prove $d$-minimality, there are two cases, depending on whether $L \neq R$ (as in Figure \ref{H-cpts}) or $L=R$ (as in Figure \ref{L=R}); in the first of these cases, we will need to consider three subcases.

\begin{figure}
\centering
\includegraphics[width = 0.9 \linewidth]{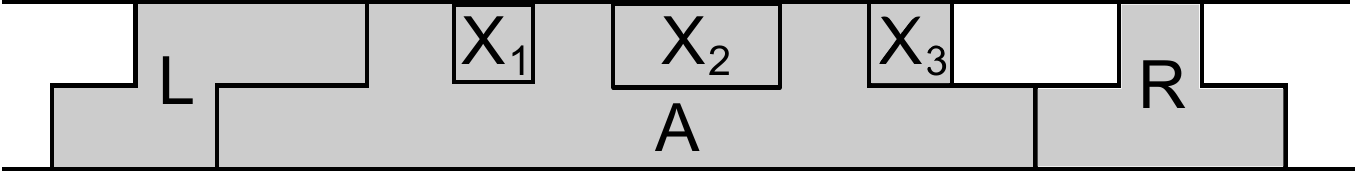}
\caption{Monochromatic components of $H$ before the final move is played.}
\label{H-cpts}
\end{figure}

\begin{figure}
\centering
\includegraphics[width = 0.7 \linewidth]{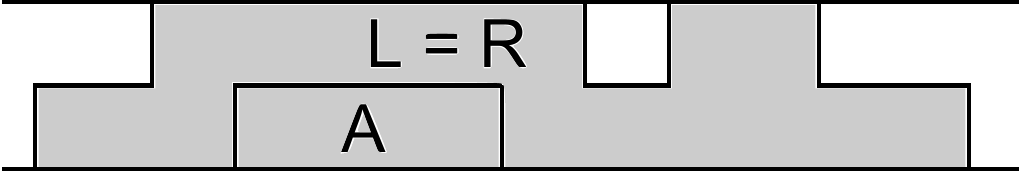}
\caption{It is possible that $L = R$.}
\label{L=R}
\end{figure}

\subsubsection*{The Construction of $T$}

Set $v$ (respectively $v'$) to be any vertex lying in the leftmost (respectively rightmost) column of $A$ that has at least one neighbour in $L$ (respectively $R$); if $L$ (respectively $R$) is empty, we set $v=u$ (respectively $v' = w$).  If two vertices of $L$ lie in the rightmost column of $L$, one of these must be adjacent to $v$, in which case we set this vertex to be $u'$; otherwise $u'$ is defined to be the unique vertex of $L$ that lies in the rightmost column.  We define $w'$ symmetrically, so that $w'$ lies in the leftmost column of $R$, and is adjacent to $v'$.  Note that, since we can flood $L$ with $d$ by playing the sequence $S$ but omitting the final move, we have $m_L(L,\omega,d) < m_H(H,\omega,d)$ and so, by the inductive hypothesis, there exists a $d$-minimal spanning tree $T_L$ for $L$ such that $\bare(T_L) \subseteq P(T_L,u,u')$.  Similarly, there exists a $d$-minimal spanning tree $T_R$ for $R$ such that $\bare(T_R) \subseteq P(T_R,w',w)$, and a $d'$-minimal spanning tree $T_A$ for $A$ such that $\bare(T_A) \subseteq P(T_A,v,v')$.  Let $S_A$ be an optimal sequence of moves to flood $T_A$ with colour $d'$, and $S_L$ and $S_R$ be optimal sequences to flood $T_L$ and $T_R$ respectively with colour $d$.  

Observe that, as well as containing vertices that lie in columns to the left (respectively right) of $A$, $L$ (respectively $R$) may additionally contain some vertices that lie in the same column as a vertex of $A$.  We set $T_L'$ to be the subtree of $T_L$ induced only by those vertices in $L$ that lie in the same column as or to the left of the leftmost vertex of $A$, and define $T_R'$ symmetrically.  Note that, even if $L=R$, we have $T_L' \cap T_R' = \emptyset$.  

Now set $T_A'$ to be the spanning tree for $H \setminus (T_L' \cup T_R')$ obtained from $T_A$ by adding an edge from every vertex $z$ of this subgraph that does \emph{not} lie in $A$ to the vertex of $A$ that lies in the same column as $z$ (and observe that $\bare(T_A') \subseteq P(T_A',v,v')$).  Assuming $T_L$ and hence $T_L'$ is nonempty, we define $x$ to be the rightmost neighbour of $v$ in $T_L'$; symmetrically, if $T_R \neq \emptyset$, we set $y$ to be the leftmost vertex of $T_R'$ that is adjacent to $v'$ (note that if $T_L' = T_L$ we will have $x = u'$, and if $T_R' = T_R$ then $y = w'$).  We then obtain a spanning tree $T$ for $H$ by connecting $T_L'$, $T_A'$ and $T_R'$ with the edges $xv$ and $v'y$.  The construction of $T$ is illustrated in Figure \ref{T-construction}.  It is clear from the construction that $T$ is a spanning tree for $H$, and that $\bare(T) \subseteq P(T,u,w)$; we will argue that in fact $T$ is a $d$-minimal spanning tree for $H$.

\begin{figure}
\centering
\includegraphics[width = 0.9 \linewidth]{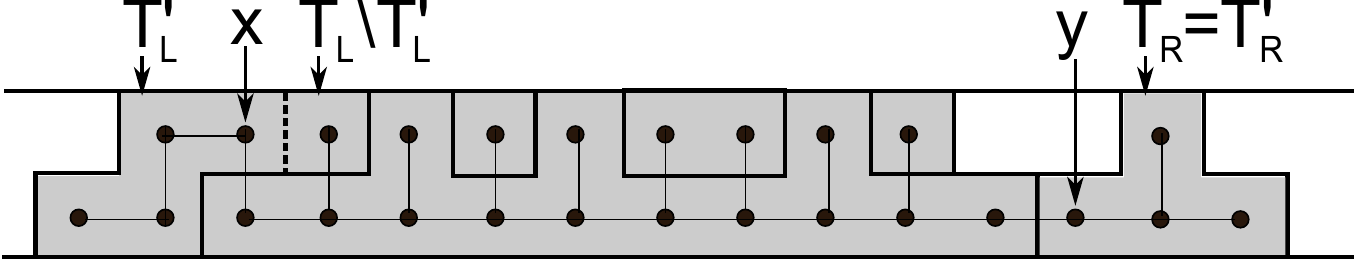}
\caption{The spanning tree $T$.}
\label{T-construction}
\end{figure}

\begin{case}
$L \neq R$.
\end{case}

Having defined the spanning tree $T$ for $H$, we now consider how to flood $T$ with colour $d$ in the case that $L \neq R$.  First, observe that
\begin{align}
|S| & \geq 1 + m_A(A,\omega,d') + m_L(L,\omega,d) + m_R(R,\omega,d) \nonumber \\
    & \qquad \qquad \qquad \qquad + \sum_{i=1}^r m_{X_i}(X_i,\omega,d) \nonumber \\
    & \geq 1 + |S_A| + |S_L| + |S_R| + | \bigcup_{i=1}^r \col(X_i,\omega) \setminus \{d\}|.
    \label{|S|}
\end{align}

We now define $S_L'$ (respectively $S_R'$) to be the subsequence of $S_L$ (respectively $S_R$) consisting of moves that change the colour of at least one vertex in $T_L'$ (respectively $T_R'$).  Observe that, as $L \neq R$ and hence $S_L' \cap S_R' = \emptyset$, we may assume without loss of generality that all moves of $S_L'$ (respectively $S_R'$) are in fact played in $T_L'$ (respectively $T_R'$): any move of $S_L'$ (respectively $S_R'$) that is not played here can be replaced with a move that is played in $T_L'$ (respectively $T_R'$) and has exactly the same effect on $T_L$ when played as part of $S_L$.  Since $T_L'$ (respectively $T_R'$) is a subtree of $T_L$ (respectively $T_R$), any sequence of moves played in $T_L'$ (respectively $T_R'$) will have the same effect on the vertices of $T_L'$ (respectively $T_R'$) as when played in the larger tree $T_L$ (respectively $T_R$); thus, as $S_L'$ (respectively $S_R'$) contains all moves of $S_L$ (respectively $S_R$) that change the colour of vertices of $T_L'$ (respectively $T_R'$), it must be that $S_L'$ (respectively $S_R'$), played in $T_L'$ (respectively $T_R'$) floods this tree with colour $d$.  This immediately implies that $m_{T_L'}(T_L',\omega,d) \leq |S_L'|$ and $m_{T_R'}(T_R',\omega,d) \leq |S_R'|$.

Let $C_L$ (respectively $C_R$) be the set of colours $\bar{d} \neq d$ such that a move of $S_L$ (respectively $S_R$) is played in a monochromatic component of colour $\bar{d}$ that does not intersect $T_L'$ (respectively $T_R'$).  We also set $C_X = \bigcup_{i=1}^r \col(X_i,\omega) \setminus \{d\}$; we will call elements of $C_A = C_X \cup C_L \cup C_R$ \emph{autonomous} colours.  Note that, for each $z \in V(T_L) \setminus V(T_L')$ that does not have colour $d$ initially, at least one of the following must hold in order for $z$ to be given colour $d$:
\begin{enumerate}
\item $\col(\{z\},\omega) \in C_L$, or
\item either initially, or after some move of $S_L$, $x$ has colour $\col(\{z\},\omega)$.
\end{enumerate}
Let $W_L$ be the set of vertices $z \in V(T_L) \setminus V(T_L')$ such that the first statement holds, so $z \in W_L$ if and only if $\col(\{z\},\omega) \in C_L$.  We then set $U_L$ to be the set of vertices in $ = V(T_L) \setminus V(T_L')$ that do not belong to $W_L$ and do not have colour $d$ initially, and note that the second statement must hold for every $v \in U_L$.  We can apply exactly the same reasoning to $V(T_R) \setminus V(T_R')$ (replacing $x$ with $y$), and define $U_R$ and $W_R$ analogously.  

Observe that $|S_L| \geq |S_L'| + |C_L|$, and $|S_R| \geq |S_R'| + |C_R|$.  Thus, by (\ref{|S|}), we see that
\begin{align}
|S| & \geq 1 + |S_A| + |S_L'| + |S_R'| + |C_L| + |C_R| + |C_X|  \label{bound-S-fine}\\
    & \geq 1 + |S_A| + |S_L'| + |S_R'| + |C_A|. \label{bound-S-coarse}
\end{align}

We will consider three sub-cases, depending on whether none, one or both of $U_L$ and $U_R$ are nonempty.  To summarise, the key properties of our construction in the case $L \neq R$ that we need when considering these subcases are as follows.
\begin{itemize}
\label{key-properties}
\item There exists a sequence $S_A$ which floods $A$ with some colour.
\item $V(T_L') \setminus V(T_L)$ (respectively $V(T_R') \setminus V(T_R)$) is partitioned into two sets $U_L$ and $W_L$ (respectively $U_R$ and $W_R$).
\item There is a sequence $S_L'$ (respectively $S_R'$) that floods $T_L'$ (respectively $T_R'$) with colour $d$ and at some point gives $x$ (respectively $y$) every colour in $\col(U_L,\omega)$ (respectively $\col(U_R,\omega)$).
\item The sets $C_X = \bigcup_{i=1}^r \col(X_i,\omega) \setminus \{d\}$, $C_L = \col(W_L,\omega)$ and $C_R = \col(W_R,\omega)$ are such that
$$|S| \geq 1 + |S_A| + |S_L'| + |S_R'| + |C_L| + |C_R| + |C_X|.$$
\end{itemize}
These properties will enable us to prove in all three subcases that $m_T(T,\omega,d) \leq |S| = m_H(H,\omega,d)$, and hence that $T$ is $d$-minimal.

\begin{subcase}
$U_L = U_R = \emptyset$.
\end{subcase}

In this case, we will first play $S_A$, flooding $A$, and then repeatedly change the colour of $A$ to cycle through all colours in $C_A$.  

First suppose that $U_L = U_R = \emptyset$.  Note in this case that our first $|S_A| + |C_A|$ moves make $T_A'$ monochromatic in some colour, so $m_{T_A'}(T_A',\omega,d) \leq 1 + |S_A| + |C_A|$.  Thus we can apply Corollary \ref{non-interference} to see that
\begin{align*}
m_T(T,\omega,d) & \leq m_{T_L'}(T_L',\omega,d) + m_{T_A'}(T_A',\omega,d) + m_{T_R'}(T_R',\omega,d) \\
                & \leq 1 + |S_A| + |S_L'| + |S_R'| + |C_A| \\
                & \leq |S| & \text{by \eqref{bound-S-coarse}}\\
                & = m_H(H,\omega,d),
\end{align*}
as required.

\begin{subcase}
Exactly one of $U_L$ and $U_R$ is nonempty.
\end{subcase}

Without loss of generality suppose that $U_L \neq \emptyset$ (and hence that $U_R = \emptyset$).  Once again, we begin by playing $S_A$ and then changing the colour of $A$ to cycle through all colours in $C_A$; these first $|S_A| + |C_A|$ moves create a monochromatic component $A'$ containing $T_A'\setminus(U_L \cup U_R)$.

We now argue that we may assume that playing these moves in $T$ does not change the colour of any vertex in $T_L'$.  First, we claim that playing $S_A$ in $T$ cannot change the colour of any vertex in $T_L'$.  Indeed, if this sequence does change the colour of a vertex in $T_L'$, it must change the colour of $x$, and this colour change will be due to moves in $S_A$ changing the colour of $v$.  Thus, if we played $S_A$ in the tree $T_1$, obtained by connecting $T_L$, $T_A$ and $T_R$ with the edges $xv$ and $yv'$, $v$ would go through the same sequence of colour changes and so the sequence would still change the colour of $x$ (which is the unique vertex of $T_L$ adjacent to $T_A$).  However, as $U_L \neq \emptyset$, we know that playing $S_L$ in $T_L$ must change the colour of $x$ (if there is a vertex in $U_R$ that does not initially have the same colour as $x$, $S_L$ must give $x$ this colour; otherwise, $x$ does not initially have colour $d$), and so by Proposition \ref{strong-non-int} (setting $X=T_A$, $Y=T_L$, $S_X = S_A$ and $S_Y = S_L$) we would have $m_{T_1}(T_1,\omega,d) \leq |S_R| + |S_L| + |S_A| < |S|$, implying (by Theorem \ref{spanning-tree}) that $m_H(H,\omega,d) \leq m_{T_1}(T_1,\omega,d) < |S| = m_H(H,\omega,d)$, a contradiction.  Hence playing $S_A$ in $T$ will not change the colour of any vertex in $T_L'$.

We may further assume that then cycling $A$ through all colours in $C_A$ does not change the colour of any vertex in $T_L'$: if $\col(\{x\},\omega) \in C_A$ we can choose this to be the last colour we play in $A$, and so our sequence will link $A$ to $x$ but will not change the colour of $x$ (or therefore of any other vertex in $T_L'$).  Thus we may indeed assume that the first $|S_A| + |C_A|$ moves do not change the colour of any vertex in $T_L'$.

Next, if playing $S_A$ and cycling through the colours of $C_A$ has not already linked $A'$ to $x$, we play one further move to give $A'$ the same colour as $x$.  Since the sequence of moves we play up to this point does not change the colour of any vertex in $T_L'$, we can now play the sequence $S_L'$ to give every vertex in $T_L'$ colour $d$.  As $x$ is in the same monochromatic component as $A'$ this will also give all vertices in $A'$ colour $d$.  Moreover, playing this sequence will at some point give $x$, and hence $A'$, every colour in $\col(U_L,\omega)$, and so will link every vertex in $U_L$ to $A'$ and ultimately give these vertices colour $d$.  Thus, playing $S_A$, cycling through $C_A$, if necessary linking $x$ to $A'$, and then playing $S_L'$ will flood all the vertices of $T \setminus T_R'$ with colour $d$ (as $U_R = \emptyset$), so we see that
$$m_{T \setminus V(T_R')}(T \setminus V(T_R'), \omega,d) \leq |S_A| + |C_A| + 1 + |S_L'|.$$
But then, once again, we can apply Corollary \ref{non-interference} to see that 
\begin{align*}
m_T(T,\omega,d) & \leq m_{T \setminus V(T_R')}(T \setminus V(T_R'), \omega,d) + m_{T_R'}(T_R',\omega,d) \\
                   & \leq 1 + |S_A| + |C_A| + |S_L'| + |S_R'| \\
                   & \leq |S| & \text{by \eqref{bound-S-coarse}}\\
                   & = m_H(H,\omega,d),
\end{align*}
as required.

\begin{subcase}
$U_L, U_R \neq \emptyset$.
\end{subcase}

In this final subcase, we begin by playing $S_A$, cycling $A$ through all colours in $C_X \cup C_L$, and then (if required) playing an additional move to change the colour of the monochromatic component containing $A$ to be the same as $x$; as before we may assume (choosing an appropriate order in which to cycle through the colours in $C_X \cup C_L$) that these initial moves do not change the colour of any vertex in $T_L'$.

Note that, as $U_L \neq \emptyset$, the colour of $x$ must change at least once when we play $S_L'$ in $T_L'$.  Set $\beta$ to be the last move in $S_L'$ to change the colour of $x$, and note then that $\beta$ must change the colour of some component $Z$, containing $x$, to $d$.  Set $\bar{T_L} = T_L' \setminus V(Z)$, and let $S_Z$ be the subsequence of $S_L'$ consisting of moves played in $Z$ (so $S_Z$ floods $Z$ with colour $d$, and $\beta$ is the final move of $S_Z$).  As $Z$ is monochromatic before $\beta$, playing $S_Z \setminus \beta$ in $Z$ must flood this component with some colour $d_Z \in C$.  Observe also that the sequence $S_L' \setminus S_Z$ must, when played in the forest $\bar{T_L}$, give every vertex of $\bar{T_L}$ colour $d$.

Suppose that, after playing $S_A$ and linking $A$ to $x$, we then play $S_Z \setminus \beta$.  This will ensure that $x$ and hence $A$ at some point receives every colour in $\col(U_L,\omega)$ (as $d \notin \col(U_L,\omega)$), so every vertex in $U_L$ is linked to $A$.  Note that we now have a monochromatic component $B$ that contains $A$, $Z$ and all vertices of $T_A' \setminus V(T_R)$ that do not initially have colour $d$.

We claim that the sequence of moves we play up to this point cannot change the colour of any vertex in $T_R'$.  To prove the validity of this claim, set $T_2$ to be the spanning tree for $H$ obtained by connecting $T_R$ and $T \setminus V(T_R)$ with the edge $yv'$.  It is clear that, if the sequence of moves we have played so far changes the colour of any vertex in $T_R'$ when played in $T$, then playing the same sequence in $T_2$ would change the colour of $y \in T_R$ ($v'$ will go through the same sequence of colour changes whether the moves are played in $T$ or $T_2$, and vertices in $T_R'$ will only change colour when the sequence is played in $T$ if the colour changes of $v'$ cause the colour of $y$ to change). However, as $U_R \neq \emptyset$, we also know that $S_R$, played in $T_R$, changes the colour of $y$.  Note also that all vertices of $T_2 \setminus (V(B) \cup V(T_R))$ that do not belong to $\bar{T_L}$ have colour $d$ initially, so $m_{\bar{T_L}}(\bar{T_L},\omega,d)$ moves suffice to flood $T_2 \setminus (V(B) \cup V(T_R))$ with colour $d$.  We can now apply Proposition \ref{strong-non-int}, setting $S_X$ to be the sequence of moves we have played up to this point, $X = B \setminus V(T_R)$, $Y = T_R$ and $S_Y = S_R$ to see that
\begin{align*}
m_{T_2}(T_2,\omega,d) & \leq m_{\bar{T_L}}(\bar{T_L},\omega,d) + |S_A| + 1 + |S_Z| - 1 + |C_X| + |C_L| + |S_R| \\
                & \leq |S_L'| - |S_Z| + |S_A| + |S_Z| + |C_X| + |C_L| + |S_R'| \\
                & = |S_L'| + |C_L| + |S_R| + |S_A| + |C_X| \\
                & < |S|.
\end{align*}
Theorem \ref{spanning-tree} would then imply that
$$m_H(H,\omega,d) \leq m_{T_2}(T_2,\omega,d) < |S| = m_H(H,\omega,d),$$
a contradiction.                              

Next we cycle the monochromatic component $B$ through all colours in $C_R$ (again, we may order these colours to ensure this does not change the colour of $y$); if this does not give $B$ the same colour as $y$, we then play one further move to link this component to $y$.  As we may therefore assume that all vertices in $T_R'$ still have their initial colouring, if we now play $S_R'$, this will flood $T_R'$ with colour $d$; as $B$ and $y$ lie in the same monochromatic component before these moves are played, this sequence will also give every vertex in $B$ colour $d$.  Moreover, linking $B$ to $y$ and playing $S_R'$ will at some point give $y$, and hence $A$, every colour in $\col(U_R,\omega)$, and so all vertices in $U_R$ will be linked to $B$ and thus end up with colour $d$.  So this sequence of moves gives every vertex in $T \setminus V(\bar{T_L})$ colour $d$, and we have
\begin{align*}
m_{T \setminus V(\bar{T_L})}(T \setminus V(\bar{T_L}), \omega,d) & \leq |S_A| + |C_X| + |C_L| + 1 + |S_Z| - 1 + |C_R| + 1 + |S_R'| \\
                                  & \leq |S_A| + |C_A| + |S_Z| + |S_R'| + 1.
\end{align*}
Finally, we apply Corollary \ref{non-interference} to give
\begin{align*}
m_T(T,\omega,d) & \leq m_{T \setminus V(\bar{T_L})}(T \setminus V(\bar{T_L}), \omega,d) + m_{\bar{T_L}}(\bar{T_L}, \omega, d) \\
              & \leq |S_A| + |C_A| + |S_Z| + |S_R'| + 1 + |S_L'| - |S_Z| \\
              & = |S_A| + |C_A| + |S_R'| + |S_L'| + 1 \\
              & \leq |S| & \text{by \eqref{bound-S-coarse}} \\
              & = m_H(H,\omega,d),
\end{align*} 
as required.  This completes the proof in the final subcase for $L \neq R$.

\begin{case}
$L = R$.
\end{case}

For the case $L=R$, the structure of $T$ is illustrated in Figure \ref{T,L=R}.  The previous reasoning fails in the case $L=R$ because $S_L = S_R$ and so we may not be able to define disjoint subsequences $S_L'$ and $S_R'$ which flood $T_L'$ and $T_R'$ respectively.  However, by considering more carefully the sequence of moves that floods $H \setminus V(A)$, we are able to deal with this problem.

\begin{figure}
\centering
\includegraphics[width=0.9 \linewidth]{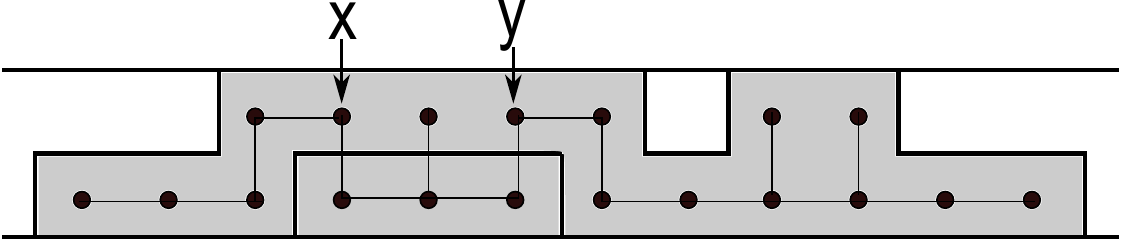}
\caption{The construction of $T$ in the case that $L=R$.}
\label{T,L=R}
\end{figure}

If $x$ and $y$ belong to the same monochromatic component $T'$ of $T_L (=T_R)$ under the initial colouring $\omega$ (where this component has colour $d_{xy}$), then we can flood $\tilde{T} = T[V(T') \cup V(A)]$ by playing $S_A$ and then changing the colour of $A$ to $d_{xy}$: this implies that $m_{\tilde{T}}(\tilde{T}, \omega, d_{xy}) \leq |S_A| + 1$.  Let $\omega'$ be the colouring of $T$ which agrees with $\omega$ on every vertex in $T_L$, and gives every vertex in $A$ colour $d_{xy}$.  Then $T$ with colouring $\omega'$ is equivalent (when monochromatic components are contracted) to $T_L$ with colouring $\omega$, implying that $m_T(T,\omega',d) = m_{T_L}(T_L,\omega,d) \leq |S_L|$.  We can then apply Lemma \ref{change-colouring} to give
\begin{align*}
m_T(T,\omega,d) & \leq m_T(T,\omega',d) + m_{T[V(T') \cup V(A)]}(T[V(T') \cup V(A)], \omega, d_{xy}) \\
                & \leq |S_L| + |S_A| + 1 \\
                & = |S| \\
                & = m_H(H,\omega,d).
\end{align*}               

So we may assume that $x$ and $y$ do not belong to the same monochromatic component initially.  Let $S'$ be the initial segment of $S_L$ up to and including the move that first links $x$ and $y$; let $T'$ be the monochromatic component of $T \setminus V(A)$ that contains $x$ and $y$ at this point, suppose that $T'$ has colour $\tilde{d}$ and that $k$ moves of $S'$ are played in $T'$.  

We now claim that it suffices to prove that 
\begin{equation}
m_{\tilde{T}}(\tilde{T},\omega,\tilde{d}) \leq |S_A| + k + 1.
\label{L=R-claim}
\end{equation}
To see that this is indeed sufficient, set $\omega'$ to be the colouring of $V(T)$ that agrees with $S'(\omega,T_L)$ on $T_L$ and gives all vertices of $A$ colour $\bar{d}$.  Note that $T$ with colouring $\omega'$ is equivalent (when monochromatic components are contracted) to $T_L$ with colouring $S'(\omega,T_L)$, and so $m_T(T,\omega',d) = m_{T_L}(T_L,S'(\omega,T_L),d) \leq |S_L| - |S'|$.  Let $\mathcal{A}$ be the set of monochromatic components of $T$ with respect to $\omega'$, and suppose that each $A \in \mathcal{A}$ has colour $c_A$ under this colouring.  As $k$ moves of $S'$ are played in $\tilde{T}$, we can bound the number of moves required to give all the other monochromatic components with respect to $\omega'$ the colour they receive under $\omega'$:
$$\sum_{\substack{A \in \mathcal{A} \\ A \neq \tilde{T}}} m_A(A,\omega,c_A) \leq |S'| - k.$$
Thus
$$\sum_{A \in \mathcal{A}} m_A(A,\omega,c_A) \leq |S'| - k + m_{\tilde{T}}(\tilde{T},\omega,\tilde{d}),$$
and so, if \eqref{L=R-claim} holds,
$$\sum_{A \in \mathcal{A}} m_A(A,\omega,c_A) \leq |S'| - k + |S_A| + k + 1 = |S'| + |S_A| + 1.$$
Lemma \ref{change-colouring} then gives
\begin{align*}
m_T(T,\omega,d) & \leq m_T(T,\omega',d) + \sum_{A \in \mathcal{A}} m_A(A,\omega,c_A) \\
                & \leq |S_L| - |S'| + |S'| + |S_A| + 1 \\
                & = |S| \\
                & = m_H(H,\omega,d),
\end{align*}
as required.  Thus it is indeed sufficient to prove \eqref{L=R-claim}.

To prove the validity of \eqref{L=R-claim}, we will invoke the reasoning used in the case $L \neq R$; to do so we must show that appropriate versions of the key properties listed on page \pageref{key-properties} hold in this case.

We begin with some definitions.  Set $\tilde{T_L}$ to be the maximal monochromatic component of $T'$ containing $x$ immediately before the final move of $S'$, and set $\tilde{T_R} = T' \setminus V(\tilde{T_L}$.  Now let $\tilde{S_L}$ (respectively $\tilde{S_R}$) be the subsequence of $S'$ consisting of moves played in $\tilde{T_L}$ (respectively $\tilde{T_R}$); note that $\tilde{S_L} \cap \tilde{S_R} = \emptyset$, $|\tilde{S_L}| + |\tilde{S_R}| = k$, and that playing $\tilde{S_L}$ (respectively $\tilde{S_R}$) in $\tilde{T_L}$ (respectively $\tilde{T_R}$) floods this tree with colour $\tilde{d}$.  Further define $\tilde{T_L'}$ (respectively $\tilde{T_R'}$) to be the subtree of $\tilde{T_L}$ (respectively $\tilde{T_R}$) induced by vertices lying in the same column as or to the left (respectively right) of $x$ (respectively $y$), and let $\tilde{S_L'}$ (respectively $\tilde{S_R'}$) be the subsequence of $\tilde{S_L}$ (respectively $\tilde{S_R}$) consisting of moves that change the colour of at least one vertex in $\tilde{T_L'}$ (respectively $\tilde{T_R'}$).  Note that, as $\tilde{S_L} \cap \tilde{S_R} = \emptyset$, we can in this case assume that every move of $\tilde{S_L'}$ (respectively $\tilde{S_R'}$) is in fact played in $\tilde{T_L'}$ (respectively $\tilde{T_R'}$), and hence that playing $\tilde{S_L'}$ in $\tilde{T_L'}$ (respectively $\tilde{S_R'}$ in $\tilde{T_R'}$) floods this subtree with colour $\tilde{d}$.

Now let $\tilde{C_L}$ (respectively $\tilde{C_R}$) be the set of colours $\bar{d} \neq \tilde{d}$ such that at least one move of $\tilde{S_L}$ (respectively $\tilde{S_R}$) is played in a monochromatic component of colour $\bar{d}$ that does not intersect $\tilde{T_L'}$ (respectively $\tilde{T_R'}$).  Note that for every vertex $z \in V(\tilde{T_L}) \setminus V(\tilde{T_L'})$ that does not initially have colour $\tilde{d}$, at least one of the following must hold:
\begin{enumerate}
\item $\col(\{z\},\omega) \in C_L$, or
\item either initially, or after some move of $\tilde{S_L}$, $x$ has colour $\col(\{z\},\omega)$.
\end{enumerate}
Let $\tilde{W_L}$ be the set of vertices $z \in V(\tilde{T_L)} \setminus V(\tilde{T_L'})$ such that the first statement holds, and set $U_L = (V(\tilde{T_L}) \setminus V(\tilde{T_L'})) \setminus \tilde{W_L}$.  Note that, for every $z \in \tilde{W_L}$, playing $\tilde{S_L'}$ in $\tilde{T_L'}$ must at some point give $x$ colour $\col(\{z\},\omega)$.  We apply exactly the same reasoning to $V(\tilde{T_R}) \setminus V(\tilde{T_R'})$ (replacing $x$ with $y$) and define $U_R$ and $W_R$ analogously.  Observe that $|\tilde{S_L}| \geq |\tilde{S_L'}| + |\tilde{C_L}|$ and $|\tilde{S_R}| \geq |\tilde{S_R'}| + |\tilde{C_R}|$.

Thus our construction has the following properties.
\begin{itemize}
\item There is a sequence $S_A$ which floods $A$ with some colour.
\item $V(\tilde{T_L}) \setminus V(\tilde{T_L'})$ (respectively $V(\tilde{T_R}) \setminus V(\tilde{T_R'})$) is partitioned into two sets, $\tilde{U_L}$ and $\tilde{W_L}$ (respectively $\tilde{U_R}$ and $\tilde{W_R}$).
\item There is a sequence $\tilde{S_L'}$ (respectively $\tilde{S_R'}$) that floods $\tilde{T_L'}$ (respectively $\tilde{T_R'}$) with colour $\tilde{d}$ and at some point gives $x$ (respectively $y$) every colour in $\col(U_L,\omega)$ (respectively $\col(U_R,\omega)$).
\item The sets $\tilde{C_X} = \bigcup_{i=1}^r \col(X_i,\omega) \setminus \{d\} = \emptyset$, $\tilde{C_L} = \col(\tilde{U_L},\omega)$ and $\tilde{C_R} = \col(\tilde{U_R},\omega)$ are such that
$$k = |\tilde{S_L}| + |\tilde{S_R}| \geq |\tilde{S_L'}| + |\tilde{C_L}| + |\tilde{S_R'}| + |\tilde{C_R}| + |\tilde{C_X}|,$$
and hence
$$|S_A| + k + 1 \geq 1 + |S_A| + |\tilde{S_L'}| + |\tilde{S_R'}| + |\tilde{C_L}| + |\tilde{C_R}| + |\tilde{C_X}|.$$
\end{itemize}
These properties correspond exactly to the list of properties on page \pageref{key-properties} that are required for the proof in the case $L \neq R$; thus we can apply the same reasoning (with the three subcases) again to show that
$$m_{\tilde{T}}(\tilde{T},\omega,\tilde{d}) \leq |S_A| + k + 1,$$
as required to demonstrate the validity of \eqref{L=R-claim}.  This completes the final case of the proof.
\end{proof}

In our analysis of the algorithm in the next section, we will need one additional result: we show in the next lemma that any tree can be flooded by an optimal sequence in which no moves are played at leaves.

\begin{lma}
Let $T$ be any tree, and $\omega$ a colouring of the vertices of $T$.  Then there exists a sequence of moves $S$, of length $m(T,\omega)$, which makes $T$ monochromatic and in which all moves are played in $\bare(T)$.
\label{no-leaf-moves}
\end{lma}

\begin{proof}
Let $S_0$ be any optimal sequence to flood $T$, and set $S_0'$ to be the subsequence of $S_0$ consisting of moves that change the colour of a vertex in $\bare(T)$.  Note that we may assume without loss of generality that all moves of $S_0'$ are played in $\bare(T)$.  Note further that $S_0 \setminus S_0'$ contains only moves played at leaves, and let $U$ be the set of leaves in which moves of $S_0 \setminus S_0'$ are played.  Observe that playing $S_0'$ in $T$ will make $T \setminus U$ monochromatic, and so we can flood the entire tree by playing a sequence $S$ which consists of $S_0'$ followed by a further $|\col(U,\omega)|$ moves, cycling through the colours still present in leaves of $T$ (playing all moves in $\bare(T)$).  Thus 
$$|S_1| \leq |S_0'| + |\col(U,\omega)| \leq |S_0'| + |U|.$$
However, it is clear that $|S_0| \geq |S_0'| + |U|$, as $S_0 \setminus S_0'$ contains at least one move played at each vertex in $U$.  Hence we see that $|S| \leq |S_0|$, and so $S$ is an optimal sequence to flood $T$ in which all moves are played in $\bare(T)$, as required.
\end{proof}

\subsection{The algorithm}
\label{algorithm}

In this section we describe our algorithm to solve $c$-\textsc{Free-Flood-It} on $2 \times n$ boards, and use results from the previous section to prove its correctness.

We begin with some further definitions.  For any section $B[b_1,b_2]$, we define $\mathcal{T}[b_1,b_2]$ to be the set of all spanning trees for $B[b_1,b_2]$.  Given any $2 \times n$ Flood-It board $B$, corresponding to a graph $G$ with colouring $\omega$ from colour-set $C$, we define a set of vectors $Z(B)$, where
\begin{align*}
Z(B) = \{(b_1,b_2, & r_1, r_2, d, I): \\
                & B[b_1,b_2] \text{ is a section}, \\
                & r_1, r_2 \in V(B[b_1,b_2]), \\
                & \exists T \in \mathcal{T}[b_1,b_2] \text{ such that } \bare(T) \subseteq P(T,r_1,r_2), \\
                & r_1 \text{ incident with } b_1, r_2 \text{ incident with } b_2, \\
                & d \in C, \\
                & I \subseteq C \}.
\end{align*}
Note that there always exists a tree $T \in \mathcal{T}[b_1,b_2]$ such that $\bare(T) \subseteq P(T,r_1,r_2)$ unless one of the following holds:
\begin{enumerate}
\item there is more than one vertex of $B[b_1,b_2]$ lying strictly to the left of $r_1$ or strictly to the right of $r_2$, or
\item there is exactly one vertex of $B[b_1,b_2]$ lying strictly to the left of $r_1$ (respectively to the right of $r_2$), which is not adjacent to $r_1$ (respectively $r_2$) and whose neighbour in the same column as $r_1$ (respectively $r_2$) has no neighbour in $B[b_1,b_2]$ other than $r_1$ (respectively $r_2$).
\end{enumerate}
Thus we can check whether this condition is satisfied in constant time.

We now introduce a function $f$ which is closely related to the minimum number of moves required to flood a $2 \times n$ board.  For any $\mathbf{z}=(b_1,b_2,r_1,r_2,d,I) \in Z(B)$ we define $f(\mathbf{z})$ to be the minimum, taken over all $T \in \mathcal{T}[b_1,b_2]$ such that $\bare(T) \subseteq P(T,r_1,r_2)$, of the number of moves that must be played in $P(T,r_1,r_2)$ to flood $P(T,r_1,r_2)$ with colour $d$, and link to $P(T,r_1,r_2)$ all leaves of $T$ that do not have colours from $I$.

It follows immediately from Lemmas \ref{leafy-path} and \ref{no-leaf-moves} that 
$$m(G,\omega) = \min_{\substack{d \in C \\ r_1 \text{ incident with } b_L \\ r_2 \text{ incident with } b_R}} f(b_L,b_R,r_1,r_2,d,\emptyset).$$

Our algorithm in fact computes recursively a function $f^*$, with the same parameters as $f$.  We will argue that, for every $\mathbf{z} \in Z(B)$, $f^*(\mathbf{z}) = f(\mathbf{z})$ and hence that it suffices to compute all values of $f^*$ in order to calculate $m(G,\omega)$.  

The first step of the algorithm is to initialise certain values of $f^*$ to zero.  We set $f^*(b_1,b_2,r_1,r_2,d,I) = 0$ if and only if, under the initial colouring, there exists a $r_1$-$r_2$ path of colour $d$ in $B[b_1,b_2]$, and all vertices in $B[b_1,b_2]$ that do not lie on this path are adjacent to the path and have colours from $I \cup \{d\}$.  All other values of $f^*(\mathbf{z})$ are initially set to infinity.  Note that, under this definition, $f^*(\mathbf{z}) = 0 \iff f(\mathbf{z}) = 0$, and that for each $\mathbf{z} \in Z(B)$ we can easily determine in time $O(n)$ whether $f^*$ should be initialised to zero or infinity.

In order to define further values of $f^*$, we introduce two more functions.  First, for any $\mathbf{z} = (b_1,b_2,r_1,r_2,d,I) \in Z(B)$, we set
$$f_1(b_1,b_2,r_1,r_2,d,I) = 1 + \min_{d' \in C} \{f^*(b_1,b_2,r_1,r_2,d',I \cup \{d\})\}.$$
We also define, for any $\mathbf{z} \in Z(B)$, 
\begin{align*}
f_2(b_1,b_2,r_1, & r_2,d,I) = \\
& \min_{\substack{(b_1,b,r_1,x_1,d,I) \in Z(B) \\ (b,b_2,x_2,r_2,d,I) \in Z(B) \\ b_1 < b < b_2 \\ x_1x_2 \in E(G)}} \{f^*(b_1,b,r_1,x_1,d,I) + f^*(b,b_2,x_2,r_2,d,I)\}.
\end{align*}
Finally, we set
$$f^*(\mathbf{z}) = \min \{f_1(\mathbf{z}), f_2(\mathbf{z})\}.$$

For the reasoning below, it will be useful to introduce another function $\theta$, taking the same parameters as $f$ and $f^*$.  For any $\mathbf{z} = (b_1,b_2,r_1,r_2,d,I) \in Z(B)$, we define
$$\theta(\mathbf{z}) = f^*(\mathbf{z}) + |B[b_1,b_2]|.$$

In the following two lemmas, we show that $f^*(\mathbf{z}) = f(\mathbf{z})$ for all $\mathbf{z} \in Z(B)$, as claimed.  We begin by demonstrating that $f^*(\mathbf{z})$ gives an upper bound for $f(\mathbf{z})$.

\begin{lma}
Let $G$ with colouring $\omega$ (from colour-set $C$) be the coloured graph corresponding to a $2 \times n$ Flood-It board $B$.  Then
$$f(\mathbf{z}) \leq f^*(\mathbf{z})$$
for all $\mathbf{z} = (b_1,b_2,r_1,r_2,d,I) \in Z(B)$.
\label{f*>=f}
\end{lma}

\begin{proof}
We proceed by induction on $\theta(\mathbf{z})$.  Recall that we have equality between $f(\mathbf{z})$ and $f^*(\mathbf{z})$ whenever $f^*(\mathbf{z}) = 0$, so certainly the base case for $\theta(\mathbf{z})=0$ must hold.  Assume therefore that $f^*(\mathbf{z}) > 0$, and that the result holds for all $\mathbf{z}'$ with $\theta(\mathbf{z}') < \theta(\mathbf{z})$.

Since $f^*(\mathbf{z}) > 0$, we must have $f^*(\mathbf{z}) \in \{f_1(\mathbf{z}),f_2(\mathbf{z})\}$.  Suppose first that $f^*(\mathbf{z}) = f_1(\mathbf{z})$.  Then, for some $d' \in C$,
\begin{align*}
f^*(b_1,b_2,r_1,r_2,d,I) & = 1 + f^*(b_1,b_2,r_1,r_2,d',I \cup \{d\}) \\
						 & \qquad \qquad \qquad \qquad \text{by definition of $f_1$} \\
                         & \geq 1 + f(b_1,b_2,r_1,r_2,d',I \cup \{d\}) \\
                         & \qquad \qquad \qquad \qquad \text{by inductive hypothesis.}
\end{align*}
But then we know, by definition of $f$, that there exists $T \in \mathcal{T}[b_1,b_2]$ and a sequence $S$ of $f(b_1,b_2,r_1,r_2,d',I \cup \{d\})$ moves, all played in $P(T,r_1,r_2)$, which, when played in $T$, floods $P(T,r_1,r_2) \supseteq \bare(T)$ with colour $d'$ and links all leaves to $P(T,r_1,r_2)$ except possibly those with colours from $I \cup \{d\}$.  By appending one further move to $S$, which changes the colour of $P(T,r_1,r_2)$ to $d$, we obtain a sequence $S'$ of length $f(b_1,b_2,r_1,r_2,d',I \cup \{d\}) + 1$ (with all moves played in $P(T,r_1,r_2)$) which, when played in $T$, floods $P(T,r_1,r_2)$ with colour $d$ and is such that all leaves of $T$ not linked to $P(T,r_1,r_2)$ by $S$ have colours from $I$.  Hence $f(b_1,b_2,r_1,r_2,d,I) \leq |S'| = 1 + f(b_1,b_2,r_1,r_2,d',I \cup \{d\})$, and so $f(b_1,b_2,r_1,r_2,d,I) \leq f^*(b_1,b_2,r_1,r_2,d,I)$, as required.

Now suppose that $f^*(\mathbf{z}) = f_2(\mathbf{z})$.  Then, by definition of $f_2$, there must exists a border $b$ with $b_1 < b < b_2$, and an edge $x_1x_2 \in E(G)$, such that $(b_1,b,r_1,x_1,d,I)$, $(b,b_2,x_2,r_2,d,I) \in Z(B)$ and 
$$f^*(b_1,b_2,r_1,r_2,d,I)  = f^*(b_1,b,r_1,x_1,d,I) + f^*(b,b_2,x_2,r_2,d,I).$$
Note that $|B[b_1,b]|$ and $|B[b,b_2]|$ are both strictly smaller than $|B[b_1,b_2]|$, so by the inductive hypothesis we have
$$f^*(b_1,b_2,r_1,r_2,d,I) \geq f(b_1,b,r_1,x_1,d,I) + f(b,b_2,x_2,r_2,d,I).$$
By definition of $f$, there exist trees $T_1 \in \mathcal{T}[b_1,b]$ and $T_2 \in \mathcal{T}[b,b_2]$, and sequences $S_1$ and $S_2$ of length $f(b_1,b,r_1,x_1,d,I)$ and $f(b,b_2,x_2,r_2,d,I)$ respectively, such that (for $i \in \{1,2\}$) all moves of $S_i$ are played in $P(T_i,r_i,x_i)$ and $S_i$ floods $P(T_i,r_i,x_i)$ with colour $d$, additionally linking all leaves of $T_i$ to $P(T_i,r_i,x_i)$ except possibly those with colours from $I$.  Now set $T = T_1 \cup T_2 \cup \{x_1x_2\}$.  It is clear that $T \in \mathcal{T}[b_1,b_2]$, and moreover that $\bare(T) \subseteq P(T,r_1,r_2)$.  Suppose $T_1'$ and $T_2'$ are the subtrees of $T_1$ and $T_2$ respectively that are given colour $d$ by $S_1$ and $S_2$, and set $T' = T_1' \cup T_2' \cup \{x_1x_2\}$.  Note that $P(T,r_1,r_2) \subseteq T'$ and that $\col(T \setminus T', \omega) = \col(T_1 \setminus T_1',\omega) \cup \col(T_2 \setminus T_2', \omega) \subseteq I$, so $f(b_1,b_2,r_1,r_2,d,I) \leq m_{T'}(T',\omega,d)$. We can then apply Corollary \ref{non-interference} to see that
\begin{align*}
f(b_1,b_2,r_1,r_2,d,I) & \leq m_{T'}(T',\omega,d) \\
                       & \leq m_{T_1'}(T_1',\omega,d) + m_{T_2'}(T_2',\omega,d) \\
                       & \leq |S_1| + |S_2| \\
                       & = f(b_1,b,r_1,x_1,d,I) + f(b,b_2,x_2,r_2,d,I) \\
                       & \leq f^*(b_1,b_2,r_1,r_2,d,I),
\end{align*}
completing the proof.
\end{proof}

Next we show that the reverse inequality also holds.

\begin{lma}
Let $G$ with colouring $\omega$ (from colour-set $C$) be the coloured graph corresponding to a $2 \times n$ Flood-It board $B$.  Then
$$f(\mathbf{z}) \geq f^*(\mathbf{z})$$
for all $\mathbf{z} = (b_1,b_2,r_1,r_2,d,I) \in Z(B)$.
\label{f*<=f}
\end{lma}
\begin{proof}
We proceed by induction on $f(\mathbf{z})$, noting again that we have equality in the base case for $f(\mathbf{z}) = 0$.  Suppose that $f(\mathbf{z}) > 0$, and that the result holds for $\mathbf{z}'$ whenever $f(\mathbf{z}') < f(\mathbf{z})$.  By definition, there exists a tree $T \in \mathcal{T}[b_1,b_2]$ and a sequence $S$ of length $f(b_1,b_2,r_1,r_2,d,I)$ such that $\bare(T) \subseteq P(T,r_1,r_2)$, all moves of $S$ are played in $P(T,r_1,r_2)$, and $S$ floods $P(T,r_1,r_2)$ with colour $d$, leaving only leaves with colours from $I$ not linked to $P(T,r_1,r_2)$.  We proceed by case analysis on $\alpha$, the final move of $S$.

Suppose first that $P(T,r_1,r_2)$ is already monochromatic before $\alpha$, and that this final move just changes its colour to $d$ from some $d' \in C$ (possibly flooding some additional leaves of colour $d$ in the process).  In this case it is clear that $f(b_1,b_2,r_1,r_2,d',I \cup \{d\}) \leq |S| - 1$ and so we can apply the inductive hypothesis to see that $f^*(b_1,b_2,r_1,r_2,d',I \cup \{d\}) \leq f(b_1,b_2,r_1,r_2,d',I \cup \{d\})$.  But then, by definition of $f_1$, we know that 
\begin{align*}
f^*(b_1,b_2,r_1,r_2,d,I) & \leq 1 + f^*(b_1,b_2,r_1,r_2,d',I \cup \{d\}) \\
                         & \leq 1 + f(b_1,b_2,r_1,r_2,d',I \cup \{d\}) \\
                         & \leq 1 + |S| - 1 \\
                         & = f(b_1,b_2,r_1,r_2,d,I),
\end{align*}
as required.

So we may assume that $P(T,r_1,r_2)$ is not monochromatic before $\alpha$: it may have either two or three monochromatic components.  Suppose first that $P(T,r_1,r_2)$ has exactly three monochromatic components before $\alpha$ is played, $A_1$, $A_2$ and $A_3$; we may assume that $A_1$ and $A_3$ have colour $d$ before $\alpha$, and that this final move gives $A_2$ colour $d$ to flood the entire path.  For $i \in \{1,2,3\}$, set $S_i$ to be the subsequence of $S \setminus \alpha$ consisting of moves played in $A_i$, and set $\bar{A_i}$ to be $A_i$ together with all leaves of $T$ that lie in the same column as a vertex of $A_i$ or whose only neighbour on $P(T,r_1,r_2)$ is in $A_i$.  Note that that $\bar{A_1}$, $\bar{A_2}$ and $\bar{A_3}$ partition the vertex set of $T$, and that $S_1$, $S_2$ and $S_3$ partition $S \setminus \alpha$.  We may assume without loss of generality that $r_1 \in A_1$ and $r_2 \in A_3$.  Observe that there must exist borders $b$ and $b'$, with $b_1 < b < b' < b_2$, such that $\bar{A_1} = B[b_1,b]$, $\bar{A_2} = B[b,b']$ and $\bar{A_3} = B[b',b_2]$.  Set $x_1x_2$ to be the edge of $T$ such that $x_1 \in A_1$, $x_2 \in A_2$ and $y_1y_2$ the edge of $T$ such that $y_1 \in A_2$ and $y_2 \in A_3$.

Note that $T[A_1] \in \mathcal{T}[b_1,b]$, $T[A_2] \in \mathcal{T}[b,b']$, and $T[A_3] \in \mathcal{T}[b',b_2]$, and moreover that we have $\bare(T[A_1]) \subseteq P(T[A_1],r_1,x_1)$, $\bare(T[A_2]) \subseteq P(T[A_2],x_2,y_1)$ and $\bare(T[A_3] \subseteq P(T[A_3],y_2,r_2)$.  Observe also that $S_1$ is a sequence of moves played in $P(T[A_1],r_1,x_1)$ that floods $P(T[A_1],r_1,x_1)$ with colour $d$ and links all leaves, except possibly those with colours from $I$, to $P(T[A_1],r_1,x_1)$, so we must have $f(b_1,b,r_1,x_1,d,I) \leq |S_1|$.  Similarly, we see that $f(b',b_2,y_2,r_2,d,I) \leq |S_3|$ and $f(b,b',x_2,y_1,d',I \cup \{d\}) \leq |S_2|$.  Since $|S_1|,|S_2|,|S_3| < |S|$, we can apply the inductive hypothesis to see that
$$f^*(b_1,b,r_1,x_1,d,I) \leq f(b_1,b,r_1,x_1,d,I) \leq |S_1|$$
and
$$f^*(b',b_2,y_2,r_2,d,I) \leq f(b',b_2,y_2,r_2,d,I) \leq |S_3|.$$
The inductive hypothesis also gives
$$f^*(b,b',x_2,y_1,d',I \cup \{d\}) \leq f(b,b',x_2,y_1,d',I \cup \{d\}) \leq |S_2|,$$
and we can then apply the definition of $f_1$ to see that
$$f^*(b,b',x_2,y_1,d,I) \leq 1 + f^*(b,b',x_2,y_1,d',I \cup \{d\}) \leq 1 + |S_2|.$$
Now we can apply the definition of $f^*$ to see that
\begin{align*}
f^*(b_1,b_2,r_1,r_2,d,I) & \leq f_2(b_1,b_2,r_1,r_2,d,I) \\
                         & \leq f^*(b_1,b,r_1,x_1,d,I) + f^*(b,b_2,x_2,r_2,d,I) \\
                         & \leq f^*(b_1,b,r_1,x_1,d,I) + f_2(b,b_2,x_2,r_2,d,I) \\
                         & \leq f^*(b_1,b,r_1,x_1,d,I) + f^*(b,b',x_2,y_1,d,I) \\
                         & \qquad \qquad \qquad \qquad + f^*(b',b_2,y_2,r_2,d,I) \\
                         & \leq 1 + |S_1| + |S_2| + |S_3| \\
                         & = |S| \\
                         & = f(b_1,b_2,r_1,r_2,d,I),
\end{align*}
as required.

For the remaining case, in which $P(T,r_1,r_2)$ has exactly two monochromatic components before $\alpha$, we can use the same reasoning as in the previous case for three components to show that we must once again have $f^*(b_1,b_2,r_1,r_2,d,I) \leq f(b_1,b_2,r_1,r_2,d,I)$, completing the proof.
\end{proof}

The final step to is to show that all values of $f^*$ can be computed in time $O(n^{10}2^{c})$.

\begin{prop}
For any $2 \times n$ Flood-It board $B$, the function $f^*(\mathbf{z})$ can be computed, for all $\mathbf{z} \in Z(B)$, in time $O(n^{10}2^{c})$.
\label{complexity}
\end{prop}
\begin{proof}
We compute values of $f^*$ recursively using a dynamic programming technique.  Our table has one entry for each pair of borders, for each possible vertex incident with each of the borders, for each colour in the colour-set and for each possible subset of colours, so the total number of entries is at most 
$$O(n^2 \cdot n^2 \cdot n \cdot n \cdot c \cdot 2^c) = O(n^6 c 2^c).$$

The table is initialised by setting all values to either zero or infinity, and for each entry we can determine which of these values it should take in time at most $O(n)$, so we can initialise the entire table in time $O(n^7 c 2^c)$.

The next step is to apply the recursive definition of $f^*$ repeatedly to all entries in the table that are not already set to zero.  Each time we apply this definition to a single entry, we take the minimum of at most $O(c + n^3)$ values (one for each choice of colour, plus one for each combination of a border and a pair of adjacent vertices on either side), each a combination of at most two other entries in the table, so each entry can be calculated in time $O(c + n^3)$.  We can therefore perform one iteration in which we apply the definition to each non-zero entry in the table in time $O(n^{9} 2^c)$.

Note that once we have initialised the table, we have the correct value of $f^*(\mathbf{z})$ for any $\mathbf{z}$ such that $\theta(\mathbf{z}) = 0$.  Moreover, the value of $f^*(\mathbf{z})$ depends only on values of $f^*(\mathbf{z}')$ where $\theta(\mathbf{z}') < \theta(\mathbf{z})$, so after $k$ iterations we will have correctly computed the value of $f^*(\mathbf{z})$ for all $\mathbf{z}$ with $\theta(\mathbf{z}) \leq k$.  Note that for every $\mathbf{z} \in Z(B)$, $\theta(\mathbf{z}) \leq 4n$, as there are at most $2n$ vertices lying between any pair of borders, and no more than $2n$ moves can be required to flood a graph with at most this many vertices, so $4n$ iterations are sufficient to guarantee we have computed all values of $f^*$ correctly.

Thus, we can compute all values of $f^*(\mathbf{z})$ for $\mathbf{z} \in Z(B)$ in time $O(n^{10} 2^{c})$, as required.
\end{proof}

We now combine the previous three results to give the proof of our main theorem.

\begin{proof}[Proof of Theorem \ref{2xn-fpt}]
Recall that, from the definition of $f$ and Lemmas \ref{leafy-path} and \ref{no-leaf-moves},  
$$m(G,\omega) = \min_{\substack{d \in C \\ r_1 \text{ incident with } b_L \\ r_2 \text{ incident with } b_R}} f(b_L,b_R,r_1,r_2,d,\emptyset).$$
Thus, in order to compute $m(G,\omega)$ in time $O(n^{10} 2^{c})$, it suffices to compute all relevant values of $f$ in time $O(n^{10} 2^{c})$.

However, we know from Lemmas \ref{f*>=f} and \ref{f*<=f} that $f(\mathbf{z}) = f^*(\mathbf{z})$ for all $\mathbf{z} \in Z(B)$, and from Proposition \ref{complexity} we know that we can compute $f^*(\mathbf{z})$ for all $\mathbf{z} \in Z(B)$ in time $O(n^{10} 2^{c})$.  This completes the proof of the theorem. 
\end{proof}

\section{\textsc{Free-Flood-It} on $2 \times n$ boards}
\label{NPhard}

In this section we prove the following theorem.
\begin{thm}
\textsc{Free-Flood-It} remains NP-hard when restricted to $2 \times n$ boards.
\label{2xnNP}
\end{thm}

This is somewhat surprising, as we have seen in the previous section that $c$-\textsc{Free-Flood-It} can be solved in polynomial time on $2 \times n$ boards, while \cite{clifford} gives a linear time algorithm to solve \textsc{Fixed Flood It} in this situation.  We demonstrate here that the problem is almost certainly not in \textbf{P} if we remove both these restrictions (that moves are always played at the same vertex, or the number or colours is bounded).  This is the first class of graphs for which such a result has been shown.

The proof is by means of a reduction from Vertex Cover, shown to be NP-hard by Karp in \cite{karp72}.  Given a graph $G=(V,E)$, we construct a $2 \times n$ Flood-It board $B_G$ as follows.

Suppose $E = \{e_1, \ldots, e_m\}$.  For each edge $e=uv \in E$ we construct the gadget $G'_e$, as illustrated in Figure \ref{gadget-G_e'}.  We will refer to the single-square components incident with the bottom edge in $G_e'$ as \emph{islands}.  $G_e'$ is then embedded in the larger gadget $G_e$, as shown in Figure \ref{gadget-G_e}.  Distinct colours $x_1^e, \ldots, x_r^e$ are used for each $e$, where $r = 2m+|V|$.  We then obtain the board $B_G$ by placing these gadgets $G_e$ in a row, as illustrated in Figure \ref{board-2xn-inf-B}.  Observe that we can take $n = m(2r+6) = 2m(2m + |V| + 3)$.  Let us also set $N = mr + 2m -1$.

\begin{figure} [h]
\centering
\includegraphics[width=0.4\linewidth]{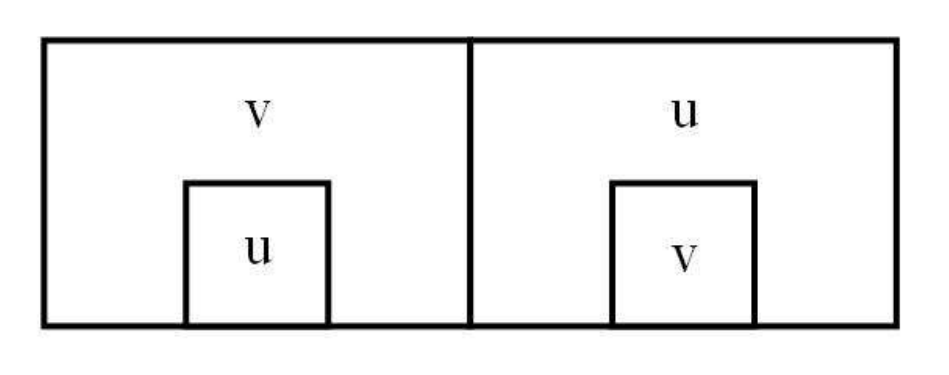}
\caption{The gadget $G_e'$}
\label{gadget-G_e'}
\end{figure}

\begin{figure} [h]
\centering
\includegraphics[width=0.8\linewidth]{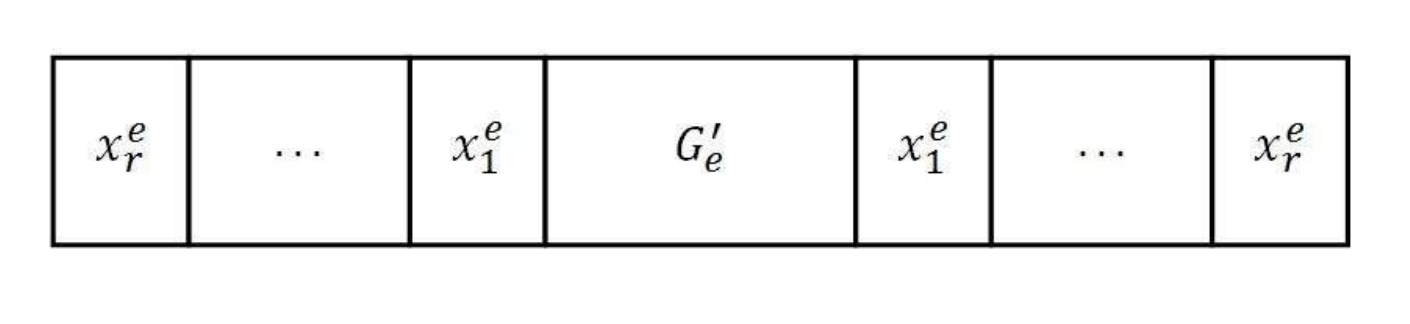}
\caption{The gadget $G_e$}
\label{gadget-G_e}
\end{figure}

\begin{figure} [h]
\centering
\includegraphics[width=0.7\linewidth]{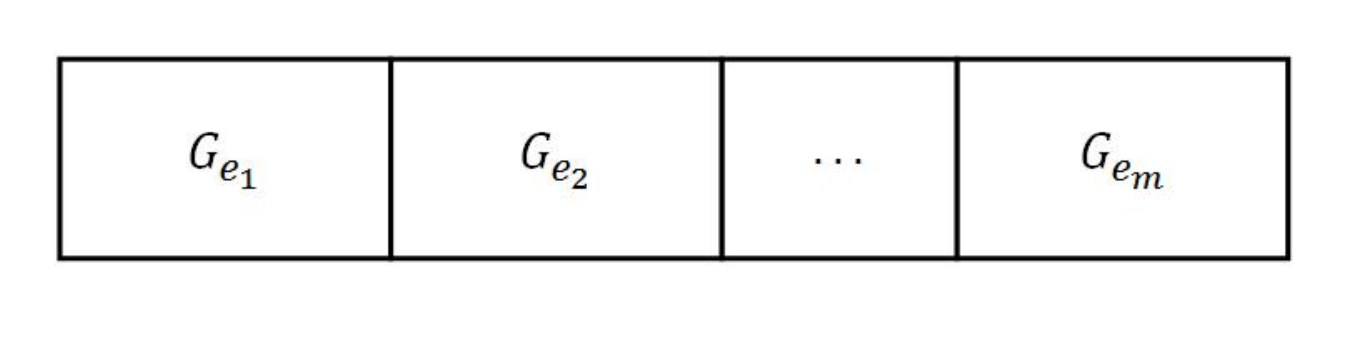}
\caption{The board B}
\label{board-2xn-inf-B}
\end{figure}

We will demonstrate that we can flood this board $B$ in $N + k$ steps if and only if $G$ has a vertex cover of size at most $k$.

\begin{lma}
If $G$ has a vertex cover of size at most $k$, then we can flood the board $B_G$ in $N + k$ steps.
\label{vc=>strat}
\end{lma}

\begin{proof}
First observe that, if $e=uv$, then with $(r+1)$ moves we can flood the gadget $G_e$, except for a single island of colour $c(e) \in \{u,v\}$, so that it is monochromatic in colour $x_r^e$: first play a single move to make all of $G_e'$ except for a single island monochromatic, then play colours $x_1^e, \ldots, x_r^e$ in this central component.  Ignoring the islands for the moment, the components corresponding to each $G_e$ now have distinct colours, so we can link these components with a minimum of $m-1$ moves.  Finally, we need to flood the islands, and this requires exactly $|\{c(e): e \in E\}|$ moves.  But we know that $G$ has a vertex cover of size at most $k$, say $V'$.  By the definition of a vertex cover, if the gadget $G_e'$ uses colours $u$ and $v$, then at least one of $u,v \in V'$.  So for each $G_e$, we may choose to leave an island of colour $d$ where $d \in V'$.  Following this strategy, we are left in the final stage with islands of at most $k$ distinct colours, and can flood these in $k$ steps (by cycling through each colour in turn in the external monochromatic component).  Hence we can flood $B_G$ in $N + k$ steps.
\end{proof}

In order to show the converse, we need one auxiliary result about optimal sequences of moves to flood trees.

\begin{lma}
Let $T$ be any tree, and $\omega$ a colouring of the vertices of $T$.  Then there exists a sequence of moves $S$, of length $m_T(T,\omega,d)$, which makes $T$ monochromatic and in which every move, except possibly the last, strictly decreases the number of maximal monochromatic components of $T$.
\label{no-pointless-moves}
\end{lma}
\begin{proof}
We proceed by induction on $m_T(T,\omega,d)$.  Note that the result is trivially true in the base case, for $m_T(T,\omega,d) = 0$.  Assume, therefore, that $m_T(T,\omega,d) > 0$ and let $S_0$ be any optimal sequence to flood $T$ with colour $d$, and that the result holds for any $T'$ with $m_{T'}(T',\omega,d) < m_T(T,\omega,d)$.  Suppose the first move of $S_0$ is $\alpha$.

First suppose that $\alpha$ strictly decreases the number of maximal monochromatic components of $T$.  Then we can apply the inductive hypothesis to see that there exists a sequence of moves $S'$ that will take $m_T(T,\alpha(\omega,T),d)$ steps to flood $T$ with colour $d$, starting from the colouring $\alpha(\omega,T)$, and such that every move of $S'$ strictly decreases the number of monochromatic components of $T$.  Thus, if we set $S$ to be the sequence of moves in which we first play $\alpha$ and then play $S'$, this sequence $S$ will have the required properties.

Now suppose that this is not the case, so $\alpha = (v,d')$ for some $v \in V(T)$ and $d' \in C$ such that no neighbours of $v$ in $T$ have colour $d'$ under $\omega$.  Note that we may assume without loss of generality that $\omega$ is a proper colouring of $T$, so $\alpha$ only changes the colour of $v$.  If in fact $V(T) = \{v\}$, then any optimal sequence to flood $T$ will have at most one move, so the result is trivially true.  Otherwise, there must be some move which links $v$ to an adjacent monochromatic component.  Set $\beta$ to be the first such move, and let $\bar{S_0}$ be the initial segment of $S_0$ up to and including $\beta$.  

Suppose that $\beta$ links $v$ to a neighbour $u$.  We then define $\alpha'=(v,\col(u,\omega))$, and claim that there exists an optimal sequence of moves to flood $T$ with colour $d$ that has $\alpha'$ as the first move; with such a sequence, we would be in the previous case and so the result would follow by the inductive hypothesis.  Thus it suffices to prove this claim, that is to show that $m_T(T,\alpha'(\omega,T),d) \leq m_T(T,\omega,d) - 1$.  

Let $A_1, \ldots, A_r$ be the monochromatic components of $T$ with respect to the colouring $\bar{S_0}(\omega,T)$, where $A_i$ has colour $d_i$ under this colouring, and assume without loss of generality that $v \in A_1$.  For $1 \leq i \leq r$, set $S_i$ to be the subsequence of $\bar{S_0}$ consisting of moves played in $A_i$; note that these subsequences partition $\bar{S_0}$ and that, for each $i$, we have $m_{A_i}(A_i,\omega,d_i) \leq |S_i|$.

By Lemma \ref{change-colouring} we see that
\begin{align*}
m_T(T,\alpha'(\omega,T),d) & \leq m_T(T,\bar{S_0}(\omega,T),d) + \sum_{i=1}^r m_{A_i}(A_i,\alpha'(\omega,T),d_i) \\
						   & \leq |S_0| - |\bar{S_0}| + \sum_{i=1}^r m_{A_i}(A_i,\alpha'(\omega,T),d_i) \\
						   & = m_T(T,\omega,d) - |\bar{S_0}| + \sum_{i=1}^r m_{A_i}(A_i,\alpha'(\omega,T),d_i),
\end{align*}
and so it suffices to prove that
\begin{equation}
\sum_{i=1}^r m_{A_i}(A_i,\alpha'(\omega,T),d_i) \leq |\bar{S_0}| - 1.
\label{pointless-claim}
\end{equation}

First observe that, for $2 \leq i \leq r$, we have $\alpha'(\omega,T)|_{A_i} = \omega|_{A_i}$, and so
$$m_{A_i}(A_i,\alpha'(\omega,T),d_i) = m_{A_i}(A_i,\omega,d_i) \leq |S_i|.$$
Now consider $A_1$, and note that $A_1$ is not monochromatic before $\beta$; suppose that, immediately before $\beta$, $A_1$ has maximal monochromatic components $A_1^1,\ldots,A_1^l$.  Recall that $v$ is not linked to any other vertex before $\beta$, so without loss of generality suppose that $A_1^1 = \{v\}$; we may further assume that $u \in A_1^2$.

We now set $S_1^j$ to be the subsequence of $S_1$ consisting of moves played in $A_1^j$; note that these subsequences partition $S_1$, and that for each $i$, $m_{A_1^j}(A_1^j,\omega,d_1) \leq |S_1^j|$.  Note further that, for $j \geq 2$, we have $\alpha'(\omega,T)|_{A_1^j} = \omega|_{A_1^j}$, and so 
$$m_{A_1^j}(A_1^j,\alpha'(\omega,T),d_1) = m_{A_1^j}(A_1^j,\omega,d_1) \leq |S_1^j|.$$
Observe also that $T[V(A_1^1 \cup A_1^2)]$ with colouring $\alpha'(\omega,T)$ is identical, after contracting monochromatic components, to $A_1^2$ with colouring $\omega$, and so we see that
$$m_{T[V(A_1^1 \cup A_1^2)]}(T[V(A_1^1 \cup A_1^2)],\alpha'(\omega,T),d_1) = m_{A_1^2}(A_1^2,\omega,d_1).$$
Then, by Corollary \ref{non-interference}, we see that
\begin{align*}
m_{A_1}(A_1,\alpha'(\omega,T),d_1) & \leq m_{T[V(A_1^1 \cup A_1^2)]}(T[V(A_1^1 \cup A_1^2)],\alpha'(\omega,T),d_1) \\
								 & \qquad \qquad \qquad \qquad + \sum_{j=3}^l m_{A_1^j}(A_1^j,\alpha'(\omega,T),d_1) \\
								 & \leq \sum_{j=2}^l m_{A_1^j}(A_1^j,\omega,d_1) \\
								 & \leq \sum_{j=2}^l |S_1^j| \\
								 & = |S_1| - |S_1^1|.
\end{align*}
Note that $\alpha \in S_1^1$, so $|S_1^1| \geq 1$, implying that in fact
$$m_{A_1}(A_1,\alpha'(\omega,T),d) \leq |S_1| - 1.$$
But then
$$\sum_{i=1}^r m_{A_i}(A_i,\alpha'(\omega,T),d_i) \leq (\sum_{i=1}^r |S_i|) - 1 = |\bar{S_0}| - 1,$$
and so \eqref{pointless-claim} holds, which completes the proof.
\end{proof}

Using this result, we now prove that the existence of a short sequence to flood $B_G$ implies the existence of a small vertex cover for $G$.

\begin{lma}
If we can flood $B_G$ in $N+k$ steps (for some $0 \leq k \leq |V|$), then $G$ has a vertex cover of size at most $k$.
\label{strat=>vc}
\end{lma}

\begin{proof}
Suppose the sequence $S$ floods $B_G$, where $|S| = N+k$.  Observe that, if we contract monochromatic components of the coloured graph corresponding to $B_G$, we obtain a tree $T$ (consisting of a path with $2m$ pendant leaf vertices); we will denote by $\omega$ the colouring this tree inherits from $B_G$.  Let $P$ be the unique path in $T$ joining the two vertices in $T$ that correspond to the monochromatic components incident with opposite ends of the board and note that, by Lemma \ref{no-leaf-moves}, we may assume that all moves of $S$ are played in $\bare(T) \subseteq P$; moreover $S$ must flood $P$ when played in this isolated path.  By Lemma \ref{no-pointless-moves} we may further assume that every move in $S$ decreases the number of monochromatic components of $T$ by at least one.

We will say that a component of colour $d$ is \emph{eliminated} by the move $\alpha$ if $\alpha$ changes the colour of that component, linking it to an adjacent component of colour $d' \neq d$.  We say that $\alpha$ \emph{eliminates the colour $d$} if it eliminates the last component of colour $d$ remaining in the graph.

We begin by observing the first move played in each gadget $G_e$ decreases the number of monochromatic components on $P$ by exactly one.  Since we assume that every move decreases the number of monochromatic components of $T$ by at least one, it follows immediately that we decrease the number of monochromatic components on $P$ by at least one unless this move links an island to the path; however, it is clear that if the first move played in any $G_e$ links an island to the path then this will in fact also reduce the number of monochromatic components on $P$ by exactly one.  So it remains to show that the first move in each $G_e$ cannot decrease the number of monochromatic components on $P$ by more than one.  Suppose that this first move changes the colour of the vertex $z$, linking it to at least one neighbour on $P$.  It is only possible for the move that eliminates $z$ to decrease the number of monochromatic components of $P$ by two if both neighbours of $z$ on $P$ have the same colour; but this is clearly not possible if both of these neighbours lie in $G_e$.  Thus it must be that $z$ has a neighbour outside $G_e$, and so $z$ has colour $x_r^e$.  But then a vertex in a different gadget $G_{\bar{e}}$ would at some point have had to have its colour changed to $x_{r-1}^e$ by some move; the first such move could not possibly decrease the number of monochromatic components of $T$, contradicting our assumption that every move of $S$ does decrease the number of monochromatic components of $T$.

We now consider the moves that link islands to $P$, and in particular set $I$ to be the subset of islands such that $v \in G_e$ belongs to $I$ if and only if $v$ is the second island in $G_e$ to be linked to $P$; set $U = \col(I,\omega)$.  By definition, for every $e=uv \in E$, we have $\{u,v\} \cap U \neq \emptyset$, and so $U$ is a vertex cover for $G$.  In the remainder of the proof we will show that we must in fact have $|U| \leq k$.

We claim that, for every $v \in U$, there exists a move which gives a component of $P$ colour $v$ but does not decrease the number of monochromatic components lying on $P$.  Note that every move that links a vertex from $I$ of colour $v$ to the path must give colour $v$ to a component of $P$, as no moves are played in leaves.  There exists at least one such move for each $v \in U$, and if for each $v$ one of these moves does not decrease the number of monochromatic components on $P$ then we are done; otherwise, there must be a move $\alpha$ which changes the colour of a component $X$ in order to link an island from $I$ of colour $v$ to the path (where $v \in G_e'$ for some $e \in E$), and which also links $X$ to some component $Y$ of $P$.  Note that $Y$ must have colour $v$ immediately before $\alpha$.  There two possibilities.
\begin{enumerate}

\item $Y$ contains a vertex $y$ which originally had colour $v$ and has never had its colour changed.  Since the other island in $G_e$ has already been linked to the path (and this must have been done by changing the colour of the island's neighbour on the path, the only other vertex in $G_e$ which initially had colour $v$), this vertex $y$ must belong to $G_{\bar{e}}'$ for some $\bar{e} \neq e$.  Without loss of generality, suppose that $G_{\bar{e}}'$ lies to the right of $G_e'$.  Let $Q$ be the segment of $P$ containing all vertices of $X$ and those vertices of $Y$ that do not lie to the right of $y$.  Suppose that the number of colours appearing in $Q$ under the initial colouring is $i$, so at least $i-1$ moves of the sequence up to and including $\alpha$ must be played on $Q$.  Note that these moves cannot have any effect on vertices lying to the right of $y$, as $y$ has not changed colour before $\alpha$ is played, so (as $X$ is a maximal monochromatic component before $\alpha$) the moves played on $Q$ up to this point do not change the colour of any vertices that do not lie on $Q$.  The number of colours in the initial colouring of $P \setminus Q$ is at least $mr - i + r$, as there are initially at least $mr$ colours in total, and colours $x_1^{\bar{e}},\ldots,x_r^{\bar{e}}$ appear both on $Q$ and on $P \setminus Q$.  All but at most one of these colours must be eliminated by moves that are either played after $\alpha$ or are not played on $Q$, so in total we have
$$|S| \geq i-1 + mr - i + r - 1 = N + |V| + 1,$$
contradicting our initial assumption that $|S| \geq N+k$ for some $k \leq |V|$.

\item Every vertex in $Y$ that initially had colour $v$ has at some point had its colour changed, so every vertex of $Y$ must have had its colour changed to $v$ by a move of played before $\alpha$.  Suppose $\beta$ was the first move that gave a monochromatic component $Z$ of $Y$ colour $v$, subject to the condition that no moves after $\beta$ and before $\alpha$ change the colour of $Z$.  If $\beta$ decreased the number of monochromatic components on $P$, then it must have linked $Z$ to an adjacent vertex $y$ that either had colour $v$ initially or was given colour $v$ by a previous move, and we must have $y \in Y$(as $Y$ is a maximal monochromatic component at a point after $\beta$ has been played).  Note that the colour of $y$ cannot change after $\beta$ and before $\alpha$, as this would mean the colour of $Z$ also changes, contradicting our choice of $\beta$.  Therefore some move before $\beta$ must have given $y$ colour $v$, since even if $y$ had colour $v$ initially it must, by assumption, change colour at some point before $\alpha$ is played.  But then there must have been a move played before $\beta$ which gave a component of $Y$ colour $v$, where this component's colour does not change again after this move and before $\alpha$: this contradicts our choice of $\beta$ as the first such move.  Hence $\beta$ must give a component of $P$ colour $v$ but does not decrease the number of monochromatic components lying on $P$, and so is the move we require.
\end{enumerate}

Thus we see that there are at least $|U|$ moves in $S$ which do not decrease the number of monochromatic components on $P$.  Then, since we know that at least one move in every $G_e$ decreases the number of monochromatic components on $P$ by exactly one, no move can decrease this number by more than two, and initially there are $2mr + 2m$ monochromatic components on $P$, we see that
$$|S| \geq mr + m + m + |U| - 1 = mr + 2m -1 + |U| = N + |U|,$$
and hence $|U| \leq k$.  So $U$ is a vertex cover of $G$ of size at most $k$, as required.
\end{proof}

\begin{proof}[Proof of Theorem \ref{2xnNP}]
The reduction from Vertex Cover is immediate from Lemmas \ref{vc=>strat} and \ref{strat=>vc}.
\end{proof}

\section{Conclusions and open problems}

We have demonstrated an algorithm which shows that the problem $c$-\textsc{Free-Flood-It}, restricted to $2 \times n$ boards, is fixed parameter tractable with parameter $c$, and on the other hand we have shown that \textsc{Free-Flood-It} remains NP-hard in this setting.  This answers an open question from \cite{clifford}, in which Clifford, Jalsenius, Montanaro and Sach showed that \textsc{Fixed-Flood-It} can be solved in time $O(n)$ on such boards.  Our results therefore give the first example of a class of graphs on which the complexity status of the fixed and free versions of the game differ.

Together with results from \cite{clifford} and \cite{general}, this almost completes the picture for the complexity of flood-filling problems restricted to $k \times n$ boards.  However, there does remain one open case:

\begin{prob}
What are the complexities of 3-\textsc{Fixed-Flood-It} and 3-\textsc{Free-Flood-It} restricted to $k \times n$ boards, in the case that $k \geq 3$ is a fixed integer?
\end{prob}

Another interesting direction for further research would be to consider extremal flood-filling problems in this setting.
\begin{prob}
What colourings of a $k \times n$ board $B$ with $c$ colours give the maximum value of $m(B)$?
\end{prob}
As a first step, it should not be hard to determine the maximum value of $m(B)$ for a $1 \times n$ board.  

Such questions can also be generalised to arbitrary graphs, leading to two more natural questions.
\begin{prob}
Given a graph $G$ and an integer $c \geq \chi(G)$, what proper colourings $\omega$ of $G$ with exactly $c$ colours maximise $m(G,\omega)$?
\end{prob}
\begin{prob}
Given a graph $G$, what proper colourings $\omega$ minimise $m(G,\omega)$?  Do such colourings necessarily use exactly $\chi(G)$ colours?
\end{prob}

\end{document}